\documentclass[11pt,twoside]{atmp}

\usepackage{amsmath,amssymb}

\usepackage[all]{xy}
\usepackage{color}

\newcommand{\ds}{\displaystyle}

\newtheorem{theorem}{Theorem} [section]
\newtheorem{lemma}[theorem]{Lemma}
\newtheorem{corollary}[theorem]{Corollary}

\newtheorem{proposition}[theorem]{Proposition}
\newtheorem{example}[theorem]{Example}

\newtheorem{remark}[theorem]{Remark}

\newcommand{\norm}[1]{||} 

\newcommand{\inn}[1]{\langle#1\rangle} 
\newcommand{\biginn}[1]{\big\langle#1\big\rangle} 

\newcommand{\dd}{{\mathrm d}}  

\newcommand{\RR}{{\bf R}}  
\newcommand{\CC}{{\bf C}}  

\newcommand{\MM}{{\bf M}}  

\newcommand{\Ll}{{\mathcal L}}  
\newcommand{\Gg}{{\mathcal G}}  
\newcommand{\Rr}{{\mathcal R}}  
\newcommand{\Ss}{{\mathcal S}}  
\newcommand{\Oo}{{\mathcal O}}  
\newcommand{\Nn}{{\mathcal N}}  

\newcommand{\Mm}{{\mathcal M}}  

\newcommand{\CP}{{\bf C}P} 
\newcommand{\HP}{{\mathbb H}P} 

\newcommand{\ii}{{\rm i}}  
\newcommand{\jj}{{\rm j}}  

\newcommand{\ra}{\rightarrow}

\newcommand{\pa}{\partial}

\renewcommand{\phi}{\varphi}
\newcommand{\la}{\lambda}

\newcommand{\na}{\nabla}
\newcommand{\al}{\alpha}
\newcommand{\be}{\beta}
\newcommand{\ga}{\gamma}
\newcommand{\Ga}{\Gamma}
\newcommand{\ve}{\varepsilon}
\newcommand{\si}{\sigma}
\newcommand{\Si}{\Sigma}

\newcommand{\om}{\omega}
\newcommand{\Om}{\Omega}

\newcommand{\wt}{\widetilde} 
\newcommand{\ov}{\overline} 
\newcommand{\wh}{\widehat} 

\renewcommand{\Re}{{\rm Re}\,} 
\newcommand{\grad}{\mbox{\rm grad}\,}
\newcommand{\Tr}{\mbox{\rm Tr}\,}
\newcommand{\Ric}{\mbox{\rm Ricci}\,}
\newcommand{\sym}{\mbox{\rm sym}\,}

\begin{document}

\title[Shear-free ray congruences]
{Shear-free ray congruences on curved space-times}


\author[P. Baird, M. Wehbe]{Paul Baird and Mohammad Wehbe}

\address{D\'epartement de Math\'ematiques \\
Universit\'e de Bretagne Occidentale \\
6 av.\ Victor Le Gorgeu -- CS 93837 \\
29238 Brest Cedex, France}  
\addressemail{Paul.Baird@univ-brest.fr, Mohammad.Wehbe@univ-brest.fr}

\begin{abstract}
A shear-free ray congruence on Minkowski space is a $3$-parameter family of null geodesics along which Lie transport of a complementary $2$-dimen\-sio\-nal spacelike subspace (called the screen space) is conformal.  Such congruences are defined by  complex analytic surfaces in the associated twistor space $\CP^3$ and are the basis of the construction of massless fields.  On a more general space-time, it is unclear how to couple the massless field with the gravitational field.  In this article we do this by considering the following Cauchy-type problem: given a Riemannian $3$-manifold $(M^3, g_0)$ endowed with a unit vector field $U_0$ that is tangent to a conformal foliation, we require that the pair extend to a space-time $(\Mm , \Gg)$ endowed with a spacelike unit vector field $U_t$ in such a way that $U_t$ simultaneously generates null geodesics and is tangent to a conformal foliation on spacelike slices $t=$ const.  
\end{abstract}

\maketitle

\tableofcontents

\section{Introduction}
\label{intro}
Twistor theory, as developed by R. Penrose and others over the last forty years, provides a fundamentally different perspective to the quantisation of physical fields, by considering the twistor space, the space of null geodesics, as the more basic object from which space-time points should be derived.  The program has had much success as well as some difficulties, particularly in its inability to adapt in any natural way, to general curved space-times.  The article \cite{Pe-3} provides background and motivation to the subject. 

One of the appealing aspects of the twistor approach, is the way complex analytic geometry comes into play.  Specifically, by what is essentially the Klein-Pl\"ucker correspondance, a light ray in Minkowsi space $\MM^4$ corresponds to a point in a $5$-dimensional CR-submanifold $\Nn^5$ of $\CP^3$.  To make the correspondance bijective, one needs to conformally compactify Minkowski space by adding a light-cone at infinity to obtain the manifold $\ov{\MM}^4$ diffeomorphic to $S^1 \times S^3$.  The picture can be unified by introducing the flag manifold ${\bf F}_{12}$ of pairs $(\ell , \Pi )$ consisting of resp. $1$- and $2$-dimensional subspaces of $\CC^4$ with $\ell \subset \Pi$ and considering the double fibration:
$$
\begin{array}{ccccc}
 & & {\bf F}_{12} & & \\
 & \swarrow & & \searrow & \\
\CP^3 & & & & G_2(\CC^4) \\
\cup & & & & \cup \\
\Nn^5 & & & & \ov{\MM}^4
\end{array}
$$ 
where $G_2(\CC^4)$ is the Grassmannian of complex $2$-dimensional subspaces of $\CC^4$ and where the left projection is given by $(\ell , \Pi ) \mapsto \ell$ and the right by $(\ell , \Pi ) \mapsto \Pi$.  

In general, a complex analytic surface $\Ss \subset \CP^3$ which intersects $\Nn^5$ defines a $3$-(real) parameter family of geodesics in a region $A\subset \MM^4$ (called a congruence) which is without shear.  That is, if $W$ represents the tangent vector field to the congruence, then at a point $x\in A$, the metric complement $W^{\bot}$ is $3$ dimensional and contains $W$ itself; if we take a $2$-dimensional spacelike complement $S$ in $W^{\bot}$, then for the congruence to be \emph{shear-free}, we require Lie transport of vectors in $S$ along $W$ to be conformal.  Such congruences can be used to generate solutions of the massless field equations.  This arises by contour integration in the twistor space and is the basis of the Penrose transform, which is an integral transform from sheaf cohomology in the twistor space into the space of massless fields, see \cite{Ea,Wa-We} for details.

If we consider Minkowski space $\MM^4$ endowed with standard coordinates $(t= x^0, x^1, x^2, x^3)$, then in \cite{Ba-Wo-5} it was shown how an SFR corresponds to a semi-conformal mapping (see Section \ref{sec:2} for the definition) on a space-like slice $\RR^3_t:= \{ (t, x^1, x^2, x^3)\in \MM^4: t \ {\rm constant}\}$.  The argument given in \cite{Ba-Wo-5} is rather indirect and depends on the analyticity of the SFR; it is essentially as follows.  An SFR defines a complex analytic surface $\Ss$ in $\CP^3$.  If we now regard $\CP^3$ as the (Hopf) bundle whose fibre at each point $x\in S^4$ is the set of almost Hermitian structures at $x$,  then $\Ss$ determines an integrable complex structure $J$ on a subset $V$ of $S^4$.  If we consider $S^4$ as being made up of two $4$-balls joined along an equatorial $3$-sphere $S^3\subset S^4$, then on endowing one of these balls $B^4$ with the Poincar\'e metric of constant curvature $-1$, by a theorem of J. C. Wood \cite{Wo}, the complex structure $J$ determines a harmonic morphism on $V$ (cf. \cite{Ba-1,Ba-Wo-6}) which extends to a semi-conformal mapping on $V \cap \pa B^4$ \cite{Ba-Wo-5}.  Provided one is careful with conventions, $\pa B^4 = S^3$ corresponds to the space-like slice $\RR^3_0$. In Section \ref{sec:2}, we give a direct argument that an SFR corresponds to an evolving family of complex-valued semi-conformal mappings on space-like slices $\RR^3_t$.  

This point of view is readily adapted to a more general construction:  we propose to replace the SFR with an evolving family of unit vector fields $U_t$ on a $3$-manifold $M^3$ which is coupled to a family of Riemannian metrics $g_t$ on $M^3$.  We require that the pair $(g_t, U_t)$ generate a space-time metric of the form $\Gg = - f\dd t^2 + g_t$, for some function $f \in C^{\infty}(I \times M^3)$ ($I \subset \RR$), in such a way that $U_t$ determines a null vector field $\pa_t + fU$ which is both geodesic and shear-free.  Here, we have to be careful about what we mean by shear-free, for there are two natural definitions which in general only coincide when $\Gg$ is the flat metric.  For the first of these, we consider a ray congruence along which Lie transport of vectors in a complementary spacelike $2$-space is conformal, i.e. neighbouring rays are without shear; this is the usual notion of a shear-free ray congruence (SFR) as discussed above in the case of Minkowski space.  

An alternative notion is to require that for each $t$, the vector field $U_t$ be tangent to a conformal foliation (or shear-free foliation) on $(M^3, g_t)$.  In order to formulate a well-posed problem, we still require that $U_t$ generate null geodesics and will refer to the corresponding congruence as one that \emph{generates conformal foliations} (CFGR: ``conformal foliation generating ray congruence").  It should be remarked that if a vector field $U$ is tangent to a conformal foliation on a domain $A$ of a Riemannian $3$-manifold $(M^3, g)$, then we can locally integrate it to a semi-conformal map $\phi : V \subset A \ra \CC$, defined up to post-composition with a conformal tansformation.  Thus the vector field $U_t$ and the map $\phi_t$ are locally interchangeable.  Furthermore, a semi-conformal map is absolutely minimizing for a natural functional, and we will see in Section \ref{sec:final} that the family $\{\phi_t\}$ of semi-conformal maps associated to the family $\{U_t\}$ satisfy a differential equation somewhat akin to the Schr\"odinger equation.

Our construction is analogous to the Cauchy problem, usually formulated to generate solutions to the Einstein equations, see \cite{Ba-Is} for an overview.  However, rather than produce an Einstein manifold, we wish to generate a space-time that supports a CFGR (or an SFR) for a given initial condition $(M^3, g_0, U_0)$, where $U_0$ is tangent to a conformal foliation on $(M^3, g_0)$.  Thus, instead of defining a field with respect to a fixed background metric, \emph{we suppose the field $\phi_t$ and the metric $g_t$ are coupled so as to produce a CFGR}. 
We are led to consider the evolution:      
\begin{equation} \label{evol-1}
\left\{ \begin{array}{llll}
{\rm (i)} & \ds\quad \frac{\pa U}{\pa t} & = &  - f\na^g_UU + \grad_gf \\
{\rm (ii)} & \ds\quad \frac{\pa g}{\pa t} & = & - 2\theta \odot (\dd f\vert_{TM}) + 2T
\end{array} \right.
\end{equation}
with the above initial conditions.
Here, $\theta = g(U, \ )$ is the $1$-form dual to $U$ and $T$ is a $t$-dependent symmetric covariant $2$ tensor on $M^3$ satisfying the property $T(U,U) = 0$ at each point.  Our key objective is to choose $f$ and $T$ in such a way that a solution produces a CFGR on the space-time $(\Mm^4 = I \times M^3, \Gg = - f^2 \dd t^2 + g_t)$.  In Theorem \ref{thm:main}, we give condition when this is the case.

We study in particular, \emph{fixed points of the evolution}, which occur when $g_0$ is conformally equivalent to a metric with respect to which the vector field $U_0$ is tangent to a conformal foliation by geodesics; the case of \emph{constant curvature}, when we can generate explicit solutions; the case when the \emph{distibution on $M^3$ orthogonal to $U_0$ is integrable}--remarkably, this property is preserved under the flow and the metric evolves according to a modified Ricci flow (Theorem \ref{thm:integrable}).  We also study a similar coupled evolution with given initial data, in order to generate a space-time supporting an SFR (Theorem \ref{thm:SFR}).  In the final section we discuss the functional representation of our construction, its relation to physical fields and conformal invariance.

\section{Shear-free ray congruences on Minkowski space-time}
\label{sec:1}
In order to describe light rays in Minkowski space, we adopt the spinor notation of Penrose \cite{Pe-1}.  Thus, to a (real) vector $(x^a) = (t=x^0, x^1, x^2, x^3) \in \MM^4$, we associate the Hermitian matrix:
$$
X = (x^{AA'}) = \frac{1}{\sqrt{2}}\left( \begin{array}{cc} t + x^1 & x^2 + \ii x^2 \\ x^2 - \ii x^3 & t - x^1 \end{array} \right) \qquad (A, A' \in \{ 0, 1\} ) 
$$
In particular, $\det X = t^2 - (x^1)^2 - (x^2)^2 - (x^3)^2$ is the Minkowski norm so that $\det X = 0$ if and only if $x^{AA'} = \xi^A\ov{\xi}^{A'}$, where to the spinor $\xi^A \in \CC^2$ we associate its conjugate $\ov{\xi}^{A'}$ whose components are given by $\ov{\xi}^{0'} = \ov{\xi^0}, \ov{\xi}^{1'} = \ov{\xi^1}$.  Indices are raised and lowered using the skew forms:
$$
(\ve_{AB}) = (\ve^{AB}) = (\ve_{A'B'}) = (\ve^{A'B'}) = \left( \begin{array}{rr} 0 & 1 \\ -1 & 0 \end{array}\right)
$$
and as usual we sum over repeated indices.

The partial derivatives $\na_{AA'}$ are defined by
$$
(\na_{AA'}) = \frac{1}{\sqrt{2}} \left( \begin{array}{cc} \pa_0 + \pa_1 & \pa_2 - \ii \pa_3 \\ \pa_2 + \ii \pa_3 & \pa_0 - \pa_1 \end{array}\right)
$$
where we write $ \pa_j = \pa / \pa x^j$, so that $\na_{AA'}x^{BB'} = \delta_A^B \delta_{A'}^{B'}$ where $\delta_A^B$ and $\delta_{A'}^{B'}$ are the Kr\"onecker delta symbols.

A shear-free ray congruence (SFR) is a congruence of null geodesics such that parallel transport of a $2$-dimensional complementary spacelike subspace (the \emph{screen space}) along each geodesic of the congruence is conformal (i.e. is without shear).  Such a congruence is described by a spinor field $\mu^A(x)$, $(x\in M^4)$ satisfying the equation
\begin{equation} \label{SFR}
\mu^A\mu^B\na_{AA'}\mu_B = 0\,.
\end{equation}
Explicitly, at each point $x \in M^4$, the vector $[\mu^0(x), \mu^1(x)]\in \CP^1 \sim S^2$ defines the direction of the light ray passing though $x$ as a point in the celestial sphere $S^2$.

The equations for an SFR are invariant by a rescaling $\mu^A(x) \mapsto \la (x)\mu^A(x)$ for some scalar function $\la$, as is apparent when we write (\ref{SFR}) in the equivalent form:
\begin{equation} \label{SFR-2}
\left\{ \begin{array}{lll}
\mu \na_{00'}\mu + \na_{10'}\mu & = & 0 \\
\mu\na_{01'}\mu + \na_{11'}\mu & = & 0
\end{array}
\right.
\end{equation}
where $\mu = \mu^0/\mu^1$.  

Let $\Nn^5\subset \CP^3$ be defined by the homogeneous equation:
\begin{equation}\label{N5}
\xi_0\ov{\eta^{0'}} + \xi_1\ov{\eta^{1'}} + \ov{\xi_0}\eta^{0'} + \ov{\xi_1}\eta^{1'}=0
\end{equation}
in terms of homogeneous coordinates $(\xi_A, \eta^{A'}) \in \CC^4$.
The correspondance between a point $[\xi_A, \eta^{A'}]\in \Nn^5$ ($\eta^{A'} \neq 0$) and points of $M^4$ are given by the incidence relation 
\begin{equation} \label{incidence}
\xi_A = \ii x_{AA'}\eta^{A'}
\end{equation}
This is a matter of linear algebra.  In general solutions to (\ref{incidence}) will be complex.  A solution $x_{AA'}$ is real if and only if the corresponding matrix $(x_{AA'})$ is Hermitian.  Under this assumption, on taking the complex conjugate of (\ref{incidence}), we obtain four equation in the four unknowns $x_{AA'}$, which we may abbreviate by $Ax = b$, where the coefficients of $A$ are given by $\eta^{A'}$ and their conjugates and $b$ is given by $\xi_{A}$ and their conjugates.  Since the determinant of $A$ vanishes, there will be either no solutions or infinitely many solutions.  By Cramer's rule, the latter case occurs if and only if the determinant of the matrix obtained from $A$ by replacing the last column by $b$, vanishes.  It is straighforward to check that this is precisely the condition (\ref{N5}).  Note that if $x_{AA'}$ is a real solution to (\ref{incidence}), then so is $x_{AA'} + \la \ov{\eta}_A\eta_{A'}$, for all $\la \in \RR$.  This is the light ray passing through $x^{AA'}$ with direction $\eta^{A'}$.  To complete the correspondence, one is required to compactify $\MM^4$ by adding a light cone at infinity to obtain the space $\ov{\MM}^4$ diffeomorphic to $S^1 \times S^3$.  Then the points $[\xi_A, 0] \in \Nn^5$ (a copy of $\CP^1$) are mapped to null rays generating this cone, see \cite{Pe-1} for details.

We can obtain $\Nn^5$ another way.  The Hopf fibration $\pi : \CP^3 \ra S^4$ is the map given by
$$
\pi ([z_1, z_2, z_3, z_4]) = [z_1 + z_2 \jj , z_3 + z_4 \jj ] \in \HP^1
$$
where for clarity, we now use homogeneous coordinates $[z_1, z_2, z_3, z_4]$ for points of $\CP^3$ and where $\HP^1$ is the quaternionic projective space.  On identifying $\HP^1$ with $S^4$ and letting $S^3$ be the equatorial $3$-sphere given by $\Re [z_1 + z_2 \jj , z_3 + z_4 z\jj ] = 0$, we see that $\pi ([z_1, z_2, z_3, z_4]) \in S^3$ if and only if 
$$
z_1 \ov{z_3} + z_2 \ov{z_4} + \ov{z_1} z_3 + \ov{z_2}z_4 = 0\,,
$$
which is once more the condition (\ref{incidence}).  Thus $\Nn^5 = \pi^{-1}(S^3)$.  It is now possible to identify $\pi^{-1}(S^3)$ with the unit tangent bundle $T^1S^3$ to $S^3$.  Then the intersection $\Ss \cap \Nn^5$ of a complex analytic surface $\Ss$ in $\CP^3$ with $\Nn^5$ determines a unit vector field on a domain of $S^3$ which is tangent to a conformal foliation (see Section \ref{sec:2} for the definition).  The basic object which defines an SFR is a complex analytic surface in $\CP^3$, as can be seen as follows.

Let $\Ss \subset \CP^3$ be a regular complex analytic surface locally parametrised in the form $(z,w) \mapsto [\xi_A(z,w), \eta^{A'}(z,w)]$, where $\xi_A$ and $\eta^{A'}$ are holomorphic in $(z,w)$.  Then, provided $\eta^{1'} \not\equiv 0$, up to a relabelling of the canonical coordinates of $\CP^3$, we can find a reparametrisation of the form $(z,w) \mapsto [\si (z,w), \tau (z,w), \beta (z), 1]$, where now $\si$ and $\tau$ are meromorphic in $(z,w)$ and $\beta$ is a meromorphic function of $z$ only.  Now the incidence relation {(\ref{incidence}) takes the form
$$
\left\{ \begin{array}{lll}
-\ii\sqrt{2} \si & = & \be u + q \\
- \ii \sqrt{2} \tau & = & \be \ov{q} + v \end{array}
\right.
$$
where we write $u = t+x^1, v = t-x^1, q = x^2 +\ii x^3$.  On eliminating the variable $w$ between these two equations, we obtain a relation of the form:
\begin{equation} \label{psi}
\psi (z, \be (z) u + q, \be (z)\ov{q} + v) = 0\,,
\end{equation}
where $\psi$ is a holomorphic function of three complex variables. Then it is quickly checked that any solution $z = z(x)$ of (\ref{psi}) determines a solution $\mu = - 1/z$ of (\ref{SFR-2}) and so defines an SFR.  Conversely, any SFR is locally given this way; this is know as the Kerr Theorem, see for example \cite{Hu-To,Pe-Ri-2}.  Solutions of the zero rest-mass field equations are now given by contour integrals of functions $\psi$ given as above. 

Explicitly, the field equations for a zero rest-mass particle of helicity $- n/2$ take the form
$$
\na^{AA'} \phi_{AB\ldots L} = 0\,.
$$
where the spinor field $\phi_{AB\ldots L}$ is symmetric in its $n$ indices.  Special solutions are given by contour integrals of the form \cite{Pe-2}:
$$
\phi_r = \frac{1}{2\pi \ii} \oint z^r f(z , u + z \ov{q}, q + zv) \dd z
$$
where $f$ is a holomorphic function of three complex variables.  Then we have the recurrence relations:
$$
\left\{ \begin{array}{ccl}  
\frac{\pa \phi_r}{\pa \ov{\zeta}} & = & \frac{\pa \phi_{r+1}}{\pa u} \\
\frac{\pa \phi_r}{\pa v} & = & \frac{\pa \phi_{r+1}}{\pa \zeta}   \qquad r = 0,1,2, \ldots
\end{array} \right.
$$   
For a non-negative integer $n$, let $\phi_{AB\cdots K}$ (with $2n$ spinor indices) be defined by
$$
\phi_0 = \phi_{00\cdots 00}, \quad  \phi_1 = \phi_{00\cdots 01}, \quad \phi_2 = \phi_{00\cdots 11}, \quad  \ldots , \quad \phi_{AB\cdots K} = \phi_{(AB\cdots K)}
$$
Then the recurrence relations above are equivalent to the spinor field equation
$$
\na^{AA^{\prime}}\phi_{AB\cdots K} = 0\,,
$$
These solutions to the zero rest-mass field equations are called algebraically special.  In their most general form, solutions are obtained via the \emph{Penrose transform}; an integral transform defined on $1$-cohomology with coefficients in the sheaf of twistor functions homogeneous of degree $-n-2$.  In general the contour in the above integral will surround a pole of $f$ and we obtain an SFR via the equation $\psi := 1/f = 0$.  

On writing $t$ for the time coordinate $x^0$ and setting $\pa_t = \pa / \pa t, \pa_1 = \pa / \pa x^1, \pa_q = \frac{1}{2}(\pa_2 - \ii \pa_3), \pa_{\ov{q}} = \frac{1}{2}(\pa_2 + \ii \pa_3)$, the equations (\ref{SFR-2}) for a SFR have an equivalent expression:
\begin{equation} \label{SFR-3}
\left\{ \begin{array}{lll}
{\rm (i)} \quad \mu \pa_1\mu - \mu^2 \pa_q\mu + \pa_{\ov{q}}\mu & = & 0 \\
{\rm (ii)} \quad \mu \pa_t\mu + \mu^2\pa_q\mu + \pa_{\ov{q}}\mu & = & 0
\end{array} \right.
\end{equation}
These are of course equations expressed in terms of local coordinates and there is some arbitrariness in their choice.  We first of all show how these two equations have an invariant expression, the first in terms of conformal foliations on space-like slices $\RR^3_t:= \{ (t, x^1, x^2, x^3)\in \MM^4 : t = {\rm const}\}$ and the second as a time evolution of such a foliation.  Foliations of this type integrate to give locally defined semi-conformal maps, which will be the basis of our adaptation of the above constructions to more general space-times. 

\section{Reformulation of the SFR equations}\label{sec:2}

A Lipschitz map $\phi : (M^m, g) \ra (N^n, h)$ between Riemannian manifolds is said to be \emph{semi-conformal} if, at each point $x\in M$ where $\phi$ is differentiable (dense by Radmacher's Theorem), the derivative $\dd \phi_x : T_xM\ra T_{\phi (x)}N$ is either the zero map or is conformal  and surjective on the complement of $\ker \dd\phi_x$ (called the \emph{horizontal} or \emph{complementary distribution}).  Thus, there exists a number $\la (x)$ (defined almost everywhere), called the \emph{dilation}, such that $\la (x)^2 g(X,Y)$ $ = $ $\phi^*h(X,Y)$, for all $X,Y\in (\ker \dd\phi_x)^{\perp}$.  If $\phi$ is of class $C^1$, then we have a useful characterisation in local coordinates, given by
$$
g^{ij}\phi_i^{\al}\phi_j^{\be} = \la^2 h^{\al\be}\,,
$$ 
 where $(x^i), (y^{\al})$ are coordinates on $M, N$, respectively and $\phi^{\al}_i = \pa (y^{\al}\circ \phi ) / \pa x^i$.  The  fibres of a smooth submersive semi-conformal map determine a conformal foliation, see \cite{Va} and conversely, with respect to a local foliated chart, we may put a conformal structure on the leaf space with respect to which the projection is a semi-conformal map.  We then have the identity:
$$
\left(\Ll_Ug\right)(X,Y) = - 2 U(\ln \la )\,g(X,Y),
$$
for $U$ tangent and $X,Y$ orthogonal to the foliation.  This latter equation can be taken to be the characterisation of a conformal foliation.  

The fundamental equation of a semi-conformal submersion $\phi : M^m \ra N^n$ relates the tension field $\tau (\phi ):= \Tr \na \dd \phi$ with the gradient of the dilation and the mean curvature $\nu$ of the fibres \cite{Ba-Wo-6}:
\begin{equation} \label{fund-scm}
\tau (\phi )  = - (n-2) \dd \phi (\grad \ln \la ) - (m-n) \dd \phi (\nu )\,.
\end{equation}
In particular, if $n = 2$, then a semi-conformal map is harmonic ($\tau (\phi ) = 0$) if and only if the fibres are minimal.  Note also that any conformal foliation only determines a corresponding (locally defined) semi-conformal map up to post-composition with a comformal mapping.  That is, if $\phi : M^m \ra N^n$ is semi-conformal and $\psi : N^n \ra P^n$ is a submersive conformal mapping, then the composition $\psi \circ \phi : M^m \ra P^n$ is also semi-conformal with the same regular fibre components as $\phi$.

In order to understand geometrically equations (\ref{SFR-3}), we first prove the following lemma.  Given a point $x$ of $\RR^3_t$, for some $t$ and a unit vector $U(x) \in T_x\RR^3_t$, we define \emph{the (future pointing) null geodesic determined by $U(x)$} to be the null ray passing through $x$ generated by the null vector $\pa_t + U(x)$.

\begin{lemma} \label{lem:Uparallel}
Let $U_0$ be a unit vector field on a domain of $\RR^3_0$, which defines a vector field $U_t$ on a domain of $\RR^3_t$ by parallel transport along null geodesics determined by $U_0$.  Then the evolution of $U$ satisfies
$$
\frac{\pa U}{\pa t} = - \na^{\RR^3_t}_UU\,.
$$
Conversely, any solution to this equation with given initial data $U\vert_{t=0} = U_0$ on $\RR^3_0$ arises this way.
\end{lemma}

\begin{proof}  By construction $U_t(y) = U_0(x)$ with $y = x + tU_0(x)$.  At a point $x \in \RR^3_0$, we have
$$
\frac{\pa U}{\pa t}\Big\vert_{t=0} = \lim_{t \ra 0} \frac{U_t(x) - U_0(x)}{t}\,.
$$
Now 
\begin{eqnarray*}
U_0(x) & = & U_t(x + t U_0(x)) \\
 & = & U_t(x) + t \dd U_t(x)(U_0(x)) + \Oo (t^2)\,,
\end{eqnarray*}
so 
\begin{eqnarray*}
\frac{\pa U}{\pa t} & = & \lim_{t \ra 0} \frac{U_0(x) - t \dd U_t(x)(U_0(x)) - U_0(x) + \Oo (t^2)}{t} \\
 & = & - \dd U_0(x)(U_0(x)) = - \na_{U_0(x)}U_0\,.
\end{eqnarray*}
The converse follows from the uniqueness of a solution to a first order partial differential equation, given initial data on a hypersurface.
\end{proof} 

Recall that a complex-valued function on a domain Minkowski space defines an SFR if and only if it satisfies equation \eqref{SFR-2}.

\begin{proposition} \label{prop:char-SFR}
Let $\mu$ define a shear-free ray congruence on a domain $W\subset \MM^4$ and let 
$$
U = \frac{1}{|\mu |^2 + 1} \left( |\mu |^2 - 1, 2\mu \right) \in S^2
$$
be the associated unit direction field.  Let $g = (\dd x^1)^2 + ( \dd x^2)^2 + ( \dd x^3)^2$ be the restriction of the canonical metric on $\MM^4$ to each spacelike slice $\RR^3_t : t =$ const.  Then on $W$:
\begin{equation} \label{char-SFR}
\left\{ \begin{array}{lrll}
{\rm (i)} \quad & \left( \Ll_Ug\right) (X,Y) & = & a g(X,Y) \\
{\rm (ii)} \quad & \frac{\pa U}{\pa t} & = & - \na^{\RR^3}_UU 
\end{array}
\right.
\end{equation}
for some function $a : W \ra \RR$, where $X,Y\in T\RR^3_t$ are vectors orthogonal to $U$ and where $\na^{\RR^3}$ denotes the Levi-Civita connection on $(\RR^3_t, g)$.  In particular, $U$ is tangent to a conformal foliation on $W\cap \RR^3_t$, for each $t$.

Conversely, any solution $U$ of (\ref{char-SFR})(i) and (ii) on a domain $W\subset \MM^4$ determines an SFR on $W$.  Furthermore, if (ii) holds on $\MM^4$ and (i) holds on an initial hypersurface $\RR^3_{t_0}$, for some $t_0$, then (i) holds for all $t$, where it is understood that the domain of $U$ may need to be restricted to avoid singular points of the associated foliation.  Thus conformality of the foliation associated to the unit vector field $U$ is preserved under the flow (\ref{char-SFR})(ii).
\end{proposition}

\begin{proof}  First note that equation (\ref{char-SFR})(i) is invariant under rescaling $U \mapsto f U$ for any real valued function $f$, so we work rather with the vector field
$$
V = \left( |\mu |^2 - 1, 2\mu \right) = (|\mu |^2 - 1) \pa_1 + 2 \mu \pa_q + \pa_{\ov{q}}
$$
where we recall that we write $x^1 = x^1, q = x^2 + \ii x^3$\,.  Now observe that (\ref{char-SFR})(i) is equivalent to the equation
$$
\left(\Ll_Vg\right)(X- \ii Y, X - \ii Y) = 0
$$
where $X,Y\in T\RR^3_t$ satisfy $||X|| = ||Y||$,  $g(X,Y)=g(X,V)=g(Y,V) = 0$, and we extend tensors to complex vectors by complex linearity.  In what follows we use the bracket $\inn{\, \cdot \,, \, \cdot \,}$ where convenient to express the inner product $g(\,\cdot \,, \, \cdot \,)$.

A complex vector with the above properties is given by
$$
X - \ii Y = \mu \pa_1 - \mu^2 \pa_q + \pa_{\ov{q}}\,.
$$
Then, on noting that $\inn{\pa_q, \pa_q} = \inn{\pa_{\ov{q}}, \pa_{\ov{q}}} = 0$ and that $\inn{\pa_q, \pa_{\ov{q}}} = \frac{1}{2}$, we have
\begin{eqnarray*}
\frac{1}{2} \left(\Ll_Vg\right)(X- \ii Y, X - \ii Y) & = & g(\na_{X- \ii Y}V, X - \ii Y) = g(V, \na_{X - \ii Y}(X- \ii Y)) \\
 & = & \biginn{(|\mu |^2 - 1)\pa_1 + 2 \mu \pa_q + 2 \ov{\mu}\pa_{\ov{q}}, \na_{(\mu \pa_1 - \mu^2 \pa_q + \pa_{\ov{q}})} (\mu \pa_1 - \mu^2 \pa_q + \pa_{\ov{q}})} \\
 & = & - (1 + |\mu |^2) \left( \mu \pa_1\mu - \mu^2 \pa_q\mu + \pa_{\ov{q}} \mu \right) \,.
\end{eqnarray*}
The equivalence between (\ref{SFR-3})(i) and (\ref{char-SFR})(i) now follows.

We perform a similar calculation in order to determine the time evolution.  Suppose that (\ref{char-SFR})(i) and (\ref{SFR-3})(ii) hold.  Then
\begin{eqnarray*}
\na^{\RR^3}_VV & = & \left\{ (|\mu |^2 - 1)\pa_1\mu + 2 \mu \pa_q \mu + 2 \ov{\mu} \pa_{\ov{q}} \mu \right\} (\ov{\mu} \pa_1 + 2 \pa_q ) \\
 &  & \qquad + \left\{ (|\mu |^2 - 1)\pa_1\ov{\mu} + 2 \mu \pa_q \ov{\mu} + 2 \ov{\mu} \pa_{\ov{q}} \ov{\mu} \right\} (\mu \pa_1 + 2 \pa_{\ov{q}} ) \\
 & = & - (|\mu |^2 + 1)\left( \pa_t\mu (\ov{\mu}\pa_1 + 2 \pa_q) + \pa_t \ov{\mu}(\mu \pa_1 + 2 \pa_{\ov{q}})\right) \,.
\end{eqnarray*}
But $\inn{\mu \pa_1 + 2 \pa_{\ov{q}}, X - \ii Y} = 0$ and $\inn{\ov{\mu}\pa_1 + 2 \pa_q, X - \ii Y} = |\mu |^2 + 1$, so that
$$
\pa_t \mu = - \frac{1}{(|\mu |^2 + 1)^2}\inn{\na_VV, X - \ii Y}\,.
$$
On the other hand
$$
\frac{\pa V}{\pa t} = (\ov{\mu}\pa_1 + 2 \pa_q)\pa_t \mu + ( \mu \pa_1 + 2 \pa_{\ov{q}})\pa_t\ov{\mu}\,,
$$
and we obtain
$$
\frac{\pa V}{\pa t} = - \frac{\inn{\na_VV, X - \ii Y}}{(|\mu |^2 + 1)^2} \left( 2(X + \ii Y) + \ov{\mu}V\right) - \frac{\inn{\na_VV, X + \ii Y}}{(|\mu |^2 + 1)^2} \left( 2(X - \ii Y) + \mu V\right)\,.
$$
On replacing $V$ with $U = V / ||V||$ and the vectors $X,Y$ with $\hat{X} = X / ||X||, \hat{Y} = Y/||Y||$, respectively, we obtain
\begin{eqnarray*}
\frac{\pa U}{\pa t} & = & - \frac{1}{2} \inn{\na_UU, \hat{X} - \ii \hat{Y}}(\hat{X} + \ii \hat{Y}) - \frac{1}{2} \inn{\na_UU, \hat{X} + \ii \hat{Y}} ( \hat{X} - \ii \hat{Y}) \\
 & = & - \na^{\RR^3}_UU\,.
\end{eqnarray*}
Conversely, if (\ref{SFR-3})(i) and (\ref{char-SFR})(ii) hold, then on reversing the above arguments, we obtain (\ref{SFR-3})(ii). 

In order to prove the last part of the Proposition, we calculate the time derivative of (\ref{char-SFR})(i), given that (\ref{char-SFR})(ii) holds.  One of our aims is to see how the pair of equations (\ref{char-SFR}) adapt to more general space times; this calculation is therefore carried out in Proposition \ref{prop:prop-conf} below, which establishes that in the flat case, on writing $\si = (\Ll_Ug)(X - \ii Y, X - \ii Y)$, where now $X, Y$ are supposed normalised to have unit length,
$$
\dd \si (\pa_t + U) = \al \si
$$
for some function $\al$.  In particular, by uniqueness of the solution of a first order PDE, if $\si$ vanishes on a hypersurface, it must vanish on the future causal space of this hypersurface.  This completes the proof of the proposition.
\end{proof}

We now wish to see how, given a family of evolving unit vector fields $U_t$ tangent to a conformal foliation, a correspond family $\{ \phi_t\}$ of semi-conformal maps evolve.  Since such a family of maps is only defined up to post-composition with a conformal transformation, we expect some gauge freedom in the choice of $\phi_t$.

\begin{corollary} \label{cor:evol-scm}  Let $U_t$ be a family of unit vector fields defined on $\RR^3_t$ evolving under the equation {\rm (\ref{char-SFR})(ii)} and let $\phi_t: \RR^3_t \ra \CC$ be a corresponding family of (locally defined) semi-conformal mappings with fibres the integral curves of $U_t$.  Then 
\begin{equation} \label{evol-scm}
d\left( \frac{\pa \phi }{\pa t} \right) (U) = -\tau (\phi )\,.
\end{equation}
Furthermore, any other semi-conformal solution of equation {\rm (\ref{evol-scm})} with the same fibres as $\phi_t$ has the form $f_t \circ \phi_t$, where $f_t$ is an \emph{arbitrary} conformal mapping of a domain of the complex plane $\CC$.
\end{corollary}

\begin{remark} The key point of the last part of the corollary, is that there is \emph{no} restriction on the conformal mapping $f_t$ imposed by \eqref{evol-scm}.
\end{remark}

\begin{proof} Differentiation of the identity $\dd \phi_t (U_t) = 0$ yields the equation
$$
\dd\left( \frac{\pa \phi }{\pa t}\right)(U) + \dd \phi \left( \frac{\pa U}{\pa t}\right) = 0\,.
$$
By the fundamental equation (\ref{fund-scm}) of a semi-conformal map, we have
\begin{eqnarray*}
\dd \phi \left( \frac{\pa U}{\pa t}\right) & = & - \dd \phi \left( \na_UU \right) \\
  & = & \tau (\phi )\,,
\end{eqnarray*}
as required. 

Let $\psi_t$ be a family of semi-conformal mappings having the same fibres as $\phi_t$ so that locally, we can set $\psi_t = f_t\circ \phi_t$ for some family of conformal functions $f_t : A \ra \CC$, where $A$ is an open subset of $\CC$.  We claim that 
\begin{equation} \label{evol-comp}
\dd\left( \frac{\pa \psi }{\pa t}\right)(U) + \tau (\psi ) = \dd f \left(\dd\left( \frac{\pa \phi }{\pa t}\right)(U) + \tau (\phi )\right) + \la^2 \tau (f) \circ \phi\,,
\end{equation} 
where $\la$ denotes the dilation of $\phi$.  In particular, if $\phi$ satisfies (\ref{evol-scm}), then $\psi$ satisfies (\ref{evol-scm}) if and only if $\tau (f) = 0$; but this is so for $f_t$ arbitrary conformal.

In order to show \eqref{evol-comp}, we calculate $\dd\left( \frac{\pa \psi}{\pa t}\right) (U)$.  First,
$$
\frac{\pa\psi}{\pa t} = \dd f \left( \frac{\pa \phi}{\pa t}\right) + \frac{\pa f}{\pa t}\circ \phi\,.
$$
Define $\Phi : I \times \RR^3_t \ra \CC$ ($I\subset \RR$) by $\Phi (t, x) = \phi_t(x)$ and write $\na^{\psi}$ for the connection in the pull-back bundle $\psi^{-1}T\CC$.  Then $\dd \left( \frac{\pa\psi}{\pa t}\right)(U) = \na^{\psi}_U\frac{\pa \psi}{\pa t}$ and we obtain
\begin{eqnarray*}
\na^{\psi}_U\left( \dd f \left( \frac{\pa \phi}{\pa t}\right)\right) & = & \na^{\psi}_U\left( \dd f \circ \dd \Phi (\pa / \pa t)\right) \\
 & = & \na^{\psi}_U\left(\dd (f\circ \Phi ) (\pa / \pa t)\right) \\
 & = & \na \dd (f \circ \Phi ) (U, \pa / \pa t) + \dd (f \circ \Phi ) (\na_U(\pa /\pa t))\,.
\end{eqnarray*}
The formula for the second fundamental form of the composition: $\na \dd (f\circ \Phi ) = \dd f (\na \dd \Phi ) + \na \dd f(\dd \Phi , \dd \Phi )$ now gives
\begin{eqnarray*}
\na^{\psi}_U \left( \dd f \left( \frac{\pa \phi}{\pa t}\right) \right) & = & \dd f\left( \na \dd \Phi (U, \pa / \pa t)\right) + \na \dd f (\dd \Phi (U), \dd \Phi (\pa / \pa t)) + \dd (f \circ \Phi ) (\na_U(\pa / \pa t)) \\
 & = & \dd f\big( \na \dd \Phi (U, \pa / \pa t) + \dd \Phi (\na_U(\pa / \pa t)\big) \quad ({\rm since}\ \dd \Phi (U) = 0)\\
 & = & \dd f \left( \na^{\Phi}_U\dd \Phi (\pa / \pa t)\right) = \dd f \left( \na^{\phi}_U \frac{\pa \phi}{\pa t}\right)\,.
\end{eqnarray*}
On the other hand, $\dd \left( \frac{\pa f}{\pa t} \circ \phi\right)(U) = 0$, so that
$$
\dd \left( \frac{\pa \psi}{\pa t}\right)(U) = \dd f \left( \dd \left(\frac{\pa \phi}{\pa t}\right)(U) \right)\,.
$$
Finally, the law for the tension field of a composition gives:
$$
\tau (\psi ) = \dd f (\tau (\phi )) + \la^2 \tau (f)\,,
$$
and \eqref{evol-comp} follows.  
\end{proof}

We can normalise the evolution $\pa \phi_t / \pa t$, at least locally, in the following way.  Let $A\subset \RR^3_0$ be an open subset on which $\phi_0$ is submersive with connected fibres.  Let $\Si \subset A$ be a slice, that is, a $2$-dimensional hypersurface which each fibre component intersects transversally in a unique point.  For example, after a suitable change of coordinates, we could take $\Si$ to be the slice: $x^1 = 0$, where we would then require $\pa \phi_0/\pa x^1 (x) \neq 0$ for all $x \in \Si \cap A$.  Points of $\Si$ would then be parametrized by the coordinates $(x^2,x^3)$.  

Let $c(s)$ be an integral curve of $U$ starting at a point of $\Si$: $c(0) \in \Si,\ c'(s) = U(c(s))$.  Then for $s_0$ sufficiently small that $c(s) \in A$ for all $0\leq s \leq s_0$, we have 
\begin{equation} \label{dphibydt}
\frac{\pa \phi }{\pa t} (c(s_0)) = - \int_{s = 0}^{s_0} \tau_{\phi}(c(s))\, \dd s\,.
\end{equation}
Note that $\tau_{\phi}(c(s)) \in T_{\phi (c(s))}\CC = T_{y_0}\CC$, where $y_0 = \phi (c(0)) = \phi (c(s))$ (since $c(s)$ parametrizes the fibre of $\phi$ through $c(0)$), hence the integral on the RHS of (\ref{dphibydt}) is well-defined and gives a vector in $T_{y_0}\CC \cong \CC$.  Of course, different choices of $\Si$ will give different values for $\pa \phi_t / \pa t$, but $\dd \left(\frac{\pa \phi_t}{\pa t}\right) (U)$ is well-defined and independent of this choice. 

\begin{example}  Let $\phi_0: \RR^3 \ra \CC$ be the semi-conformal map $\phi_0(x^1, x^2, x^3) = \sqrt{(x^1)^2 + (x^2)^2} + \ii x^3$.  Then the fibres of $\phi_0$ are circles whose centres lie on the $x^3$-axis lying in planes parallel to the $x^1x^2$-plane.  On letting $U_0$ be the unit vector field tangent to the fibres of $\phi_0$, we have
$$
U_0 = \left( \frac{x^2}{\sqrt{(x^1)^2 + (x^2)^2} }, \frac{-x^1}{\sqrt{(x^1)^2 + (x^2)^2} }, 0\right)\,.
$$
Take $\Si$ to be the slice $x^1 = 0$; then an integral curve of $U_0$ through $(0, x^2, x^3)$ is given by 
$$
c(s) = \left( x^2 \sin \left(\frac{s}{x^2}\right) , x^2 \cos \left(\frac{s}{x^2}\right) , x^3\right)\,.
$$
Now $\tau_{\phi_0} = 1 / \sqrt{(x^1)^2 + (x^2)^2}$, so that
$$
\frac{\pa \phi_0}{\pa t} = - \int_{s = 0}^{s_0} \tau_{\phi_0}(c(s)) \, \dd s = - \int_0^{s_0} \frac{1}{x^2} \dd s = - \frac{s_0}{x^2} = - \arg (x^1+ \ii x^2)\,.
$$
In fact it is possible to calculate the complete evolution of $\phi_0$ to be
$$
\phi_t = \ii x^3 + r - t \arg \left(\frac{r - \ii t}{x^1 - \ii x^2}\right) \,,
$$
where $r = \sqrt{(x^1)^2 + (x^2)^2 - t^2}$; the map $\phi_t$ is defined and smooth on the domain $(x^1)^2 + (x^2)^2 > t^2$.
\end{example}   

\section{Coupled evolution of a metric and a unit vector field}  Consider a triple $(M^3, g_t, U_t)$ consisting of a $3$-manifold $M^3$ endowed with an evolving family of Riemannian metrics $g_t$, together with an evolving unit vector field $U_t$.  We suppose that the parameter $t$ belongs to some connected interval $I\subset \RR$ and we consider the space-time $I\times M^3$ endowed with a metric $\Gg = - f(t,x)^2 \dd t^2 + g_t$, where $f : I \times M^3 \ra \RR$ is a non-vanishing smooth function.  Then on $I\times M^3$, the vector field $\pa_t + fU$ is null; we will call its integral curves the \emph{null curves generated by $U$}. Our aim is to determine conditions when these are geodesic as well as \emph{either} shear-free \emph{or} generate conformal foliations. 

We first determine some basic quantities deriving from the metric $\Gg$. In what follows, we fix a system of local coordinates $(x^0 = t, x^1, x^2, x^3)$ and use the following range of indices: $i,j, \ldots \in \{ 1,2,3\}; \ a,b, \ldots \in \{ 0,1,2,3\}$.   Recall that the Christoffel symbols are given by
$$
\Ga^c_{ab} = \frac{1}{2}\Gg^{cd}\left( \pa_a\Gg_{bd} + \pa_b\Gg_{ad} - \pa_d\Gg_{ab}\right)
$$
and the Riemannian curvature by
$$
R^d_{abc} = \pa_a \Ga^d_{bc} - \pa_b\Ga^l_{ac} + \Ga^e_{bc}\Ga^d_{ae} - \Ga^e_{ac}\Ga^d_{be}\,,
$$
where we write $R(\pa_a, \pa_b)\pa_c = \na_{\pa_a} \na_{\pa_b}\pa_c - \na_{\pa_b}\na_{\pa_a}\pa_c = R^d_{abc} \pa_d$.  We also use the convention $\Gg (R(\pa_a, \pa_b)\pa_c, \pa_d)= R_{abcd} = g_{de}R^e_{abc}$.  We will also have occasion to consider the Christoffel symbols and curvature of the Riemannian metric $g_t$; when there is the possibility of confusion, we will write ${}^{\Gg}\Ga^k_{ij}$ or ${}^g\Ga^k_{ij}$ to distinguish the two cases; similarly for the curvature and the connection. 

\begin{lemma} \label{lem:christoff-Gg}  With the notational conventions for indices given above, the Christoffel symbols of the metric $\Gg= - f(t,x)^2 \dd t^2 + g_t$ are given as follows:
\begin{eqnarray*}
\Ga^0_{00} = \pa_t \ln f \,, \quad \Ga^0_{0i} = \pa_i\ln f\,, \quad \Ga^i_{00} = fg^{ij}\pa_jf\,, \\
\Ga^k_{0i} = \frac{1}{2}g^{kj}\pa_tg_{ij} \,, \quad \Ga^0_{ij} = \frac{1}{2f^2}\pa_t g_{ij}\,, \quad {}^{\Gg}\Ga^k_{ij} = {}^g\Ga^k_{ij} 
\end{eqnarray*}
\end{lemma}

\begin{lemma} \label{lem:curv-Gg}  With the notational convention for indices given above and writing $\na \dd f$ for the second fundamental form of the function $f$ considered as a function on $M^3$ (with $t$ fixed), the components of the Riemannian curvature tensor of $\Gg$ are given as follows:
$$
\begin{array}{l}
R_{0i0j} = - \frac{1}{4}g^{kl}\pa_tg_{ik} \pa_tg_{jl} + \frac{1}{2}\pa_t(\pa_tg_{ij}) - f\na \dd f (\pa_i, \pa_j) - \frac{1}{2}\pa_t\ln f \pa_t g_{ij} \\
R^0_{ijk} = \frac{1}{2f^2}\left(  g_{lj}\pa_t\Ga^l_{ik} - g_{li}\pa_t\Ga^l_{jk} - \pa_i(\ln f)\pa_tg_{jk} + \pa_j(\ln f)\pa_tg_{ik}   \right) \\
{}^{\Gg}R^l_{ijk} = {}^gR^l_{ijk} + \frac{1}{4f^2}\big( \pa_t g_{jk}\pa_tg_{il} - \pa_tg_{ik}\pa_tg_{jl}\big)\,.
\end{array}
$$
\end{lemma}

\begin{proof}  In fact these follow from standard identities, the first known as the Mainardi equation, the second the Codazzi identity and the third, the Gauss equation, see \cite{Ba-Is}.  We will prove the second of the above formulae, the other two being routine calculations.  We have
\begin{eqnarray*}
R^0_{ijk}  & = & \pa_i\Ga^0_{jk} - \pa_j\Ga^0_{ik} + \Ga^0_{jk}\Ga^0_{i0} - \Ga^0_{ik} \Ga^0_{j0} + \Ga^l_{jk}\Ga^0_{il} - \Ga^l_{ik}\Ga^0_{jl} \\ 
 & = & \frac{1}{2f^2}\pa_t(\pa_ig_{jk} - \pa_jg_{ik}) - \frac{1}{2f^2}\left( \pa_i(\ln f)\pa_tg_{jk} - \pa_j(\ln f)\pa_tg_{ik}\right) \\
 & & \qquad \qquad  + \frac{1}{2f^2}\left( \Ga^l_{jk}\pa_tg_{il} - \Ga^l_{ik}\pa_tg_{jl}\right) \,.
\end{eqnarray*}
However, note that
$$
\pa_ig_{jk} - \pa_jg_{ik} = 2g_{lj}\Ga^l_{ik} + \pa_kg_{ij} = - 2g_{li}\Ga^l_{jk} - \pa_kg_{ij}
$$
so that
$$
\pa_ig_{jk} - \pa_jg_{ik} = g_{lj}\Ga^l_{ik} - g_{li}\Ga^l_{jk}\,.
$$
Substitution of this equality into the above expression now gives the formula of the lemma.
\end{proof}

From Lemma \ref{lem:christoff-Gg}, we have the following consequences. 

\begin{corollary} \label{cor:cov-der-U}  The covariant derivative of $U$ and $\pa_t + fU$ along null curves defined by $U$ is given respectively by the expressions:
$$
\begin{array}{lll}
{\rm (i)}  \ds \quad  \na^{\Gg}_{\pa_t + fU}U & = & \ds \frac{\pa U}{\pa t} + f\na^g_UU + \frac{1}{2}g^{kl}(\pa_tg_{il})u^i\pa_k \\
 & & \qquad  + \left( \dd\ln f(U) + \frac{1}{2f}u^iu^j\pa_tg_{ij}\right) \pa_t\\
{\rm (ii)}  \ds  \quad \na^{\Gg}_{\pa_t + fU}(\pa_t + fU) & = & f \grad_gf + (\pa_t + fU)(f) U \\
 & & \qquad \ds + \frac{f}{2}u^ig^{jk}\pa_tg_{ij} \pa_k  + (\pa_t\ln f + U(f))\pa_t\,.
\end{array} 
$$
where, on writing $U = u^i(t, x^k)\pa_i$ with respect to (fixed) local coordinates $(x^k)$ on $M$, we set $\pa U/\pa t = (\pa u^i/\pa t)\pa_i$\,.  
\end{corollary}

In what follows, we suppose that $T = T_t$ is a symmetric covariant $2$-tensor defined on each $(M^3, g_t)$ satisfying $T(U,U) = 0$, which, along with the function $f$ is to be determined.  In view of Corollary \ref{cor:cov-der-U}(i), in order to preserve the length of $U$, we introduce the coupled evolution:
\begin{equation} \label{coupled-evol-0}
\left\{ \begin{array}{lll}
{\rm (i)} \qquad \ds \frac{\pa U}{\pa t} & = & - f\na^g_UU - \frac{1}{2}g^{kl}(\pa_tg_{il})u^i \pa_k \\
{\rm (ii)} \qquad \ds\frac{\pa g}{\pa t} & = & - 2 \,\theta \odot \left(\dd f\vert_{TM}\right) + 2T
\end{array}\right.
\end{equation}
where $\theta = \theta (t)$ is the $1$-form dual to $U$: $\theta = g(U, \ \cdot \ )$, $\al \odot \be$ is the symmetric product of $1$-forms: $\al \odot \be = \frac{1}{2}(\al \otimes \be + \be \otimes \al )$ and $U = u^i \pa_i$, where $(x^i)$ is a system of local coordinates on $M^3$.  Note that in the case of the flat metric, when we take $f \equiv 1$ and $T \equiv 0$, we recover the evolution (\ref{char-SFR})(ii).  If we substitute into (i) the expression for $\pa g/\pa t$ given by (ii), we obtain the alternative form:
\begin{equation} \label{coupled-evol}
\left\{ \begin{array}{lll}
{\rm (i)} \qquad \ds \frac{\pa U}{\pa t} & = & - f\na^g_UU + \frac{1}{2}\grad_gf + \frac{1}{2}U(f)U - \wh{T}(U) \\
{\rm (ii)} \qquad \ds\frac{\pa g}{\pa t} & = & - 2 \,\theta \odot \left(\dd f\vert_{TM}\right) + 2T
\end{array}\right.
\end{equation}
where $\wh{T}$ is the linear transformation defined by $g(\wh{T}(E), F) = T(E,F)$ at each point $(t, x)$ for all $E,F \in T_xM^3$.

In general, the null vector field $\pa_t + fU$ is no longer geodesic, the expression for $\ds \na_{\pa_t + fU} (\pa_t + fU)$ being given by Corollary \ref{cor:cov-der-U}, however, the following lemma gives the condition when it is geodesic.    

\begin{lemma} \label{lem:props-U}  Let the metric $g_t$ and the vector field $U_t$ evolve according to equations {\rm (\ref{coupled-evol})}.  Then

{\rm (i)}  $(\pa_t + fU) g(U,U) = 0$, in particular, if $U$ is of unit length on the hypersurface $t = 0$, then it remains so for all time;

{\rm (ii)} $ \ds \na^{\Gg}_{\pa_t + fU}U = 0$, in particular $U$ is parallel along null geodesics determined by $U$;

{\rm (iii)}  $(\pa_t + fU)\Gg (U, \pa_t) = 0$, in particular, if $U$ is tangent to $M^3$ for $t=0$, then it remains so for all time;

{\rm (iv)} The null field $\pa_t + fU$ is geodesic if and only if 
$$
\wh{T}(U) = - \frac{1}{2}\grad_gf + \frac{1}{2}U(f)U\,,
$$
where $\wh{T}$ is the linear mapping determined at each point $x$ by the identity $g(\wh{T}(E), F) = T(E,F)$ for all $E,F \in T_xM^3$.
\end{lemma}

\begin{proof}  The first part of the Lemma follows directly from (\ref{coupled-evol-0})(i) and the second from Corollary \ref{cor:cov-der-U}(i).  To show part (iii), it suffices to note that
$$
\na^{\Gg}_{\pa_t + fU}\pa_t = (\pa_t\ln f + U(f))\pa_t + f\grad_gf + \frac{f}{2}u^ig^{kj}\pa_tg_{ij}\pa_k\,,
$$
so that, from (ii),
$$
(\pa_t + fU)\Gg (U, \pa_t) = \Gg (U, \na^{\Gg}_{\pa_t + fU}\pa_t) = fU(f) + \frac{f}{2}(\pa_tg)(U,U) = 0\,.
$$  
For part (iv), the field $\pa_t + fU$ is geodesic if and only if 
$$
\na_{\pa_t + fU}(\pa_t + fU) = \al (\pa_t + fU)\,,
$$
for some function $\al$.  By Corollary \ref{cor:cov-der-U}(ii),
\begin{eqnarray*}
\na^{\Gg}_{\pa_t + fU}(\pa_t + fU) & = & f \grad_gf + (\pa_t + fU)(f) U 
  + \frac{f}{2}u^ig^{jk}\pa_tg_{ij} \pa_k  + (\pa_t\ln f + U(f))\pa_t \\
 & = & \left\{ \pa_t \ln f + U(f)\right\} (\pa_t + fU) + f\grad_g f \\
 & & \qquad  + \frac{f}{2}u^ig^{kj} ( - \theta_i \pa_jf - \theta_j\pa_if + 2T_{ij}) \pa_k \\
 & = & \left\{ \pa_t \ln f + U(f)\right\} (\pa_t + fU) + \frac{f}{2}\grad_g f - \frac{f}{2}U(f)U + f\wh{T}(U)\,.
\end{eqnarray*}
Since the last three terms are tangent to $M^3$ (there is no $\pa_t$-component), the geodesic condition is equivalent to the vanishing of their sum, which gives the equation of (iii). 
\end{proof}

Let $(M^3, g_t, U_t)$ evolve according to (\ref{coupled-evol}) and let $\{ X_t, Y_t\}$ be an evolving orthonormal frame tangent to $M^3$ and orthogonal to $U_t$, which we call a \emph{complementary frame}.  Write $Z = X - \ii Y$ and let $\si =\si (t,x)$ and $\rho = \rho (t,x)$ $((t,x)\in I \times M^3)$ be the quantities defined by:
$$
\begin{array}{l}
\si = \frac{1}{2}(\Ll_Ug)(Z, Z) = g(\na^g_ZU, Z)  =  \Gg (\na^{\Gg}_ZU, Z) \\
\rho = g(\na^g_{\ov{Z}}U, Z)  =  \Gg (\na^{\Gg}_{\ov{Z}}U, Z)\,.
\end{array}
$$
The first of these is called the \emph{shear} (see \cite{Pe-Ri-2}) and measures the discrepancy from being conformal of the foliation determined by $U$.  It is quadratic in $Z$ and so if we choose an alternative complex frame $Z' = e^{\ii \theta}Z$, then the shear is now given by $\si ' = e^{2\ii\theta} \si$, so that although $\si$ is not an invariant quantity, the property of its vanishing is.  When necessary, we shall write $\si (Z, Z)$ to indicate the dependence on $Z$.   On the other hand, $\rho$ is invariant independent of the choice of $Z$.  Specifically,
$$
g(\na^g_{\ov{Z}}U, Z) = - g(U, \na^g_XX + \na^g_YY - \ii [X, Y])
$$
where $g(U, \na^g_XX + \na^g_YY)$ is the mean curvature of the distribution on $M^3$ orthogonal to the vector field $U$, and where $g(U, [X, Y])$ measures the integrability of this distribution. 

We are mainly interested in the evolution of the shear $\si$, however, we will find that both quantities $\si$ and $\rho$ are intimately related and so we need to calculate both of their evolutions under the flow (\ref{coupled-evol}).  Before doing this, we establish some basic formulae that will be useful in the computations that follow.  In what follows, we assume that $g_t$ and $U_t$ evolve according to (\ref{coupled-evol}).  The following lemma is a simple calculation involving the Christoffel symbols as given by Lemma \ref{lem:christoff-Gg}.

\begin{lemma}  \label{lem:comp}  {\rm (i)}  For a vector field $E = e^i\pa_i$ tangent to $M^3$, we have
$$
\na_E\pa_t = E(\ln f) \pa_t - \frac{f}{2}g(E, U) \grad_g \ln f - \frac{f}{2}E(\ln f) U + e^i T_{ij}g^{kj}\pa_k\,;
$$

{\rm (ii)}  
$$
\Gg (\na^{\Gg}_{\pa_t}Z, \ov{Z}) = T(Z, \ov{Z})\,;
$$
\end{lemma}

There are two concepts of shear-free ray congruence that we can consider on the space-time $\Mm^4 = I \times M^3$ endowed with the metric $\Gg$, both of which coincide in the flat case.  In both cases, we suppose that the metric $g_t$ and the unit vector field $U_t$ evolve according to (\ref{coupled-evol}) with the condition of Lemma \ref{lem:props-U}(iv) satisfied, so the null curves generated by $U$ are geodesic.  The first notion requires that Lie transport of vectors in the `screen space', that is the orthogonal complement of $U$ in $TM^3$, along the light rays be conformal; as in the flat case, we shall call this a \emph{shear-free ray congruence} (SFR).  The second concept requires that the evolution (\ref{coupled-evol}) \emph{preserves the property that the vector field $U$ be tangent to a conformal foliation}; we shall refer to the corresponding ray congruence as one that generates conformal foliations (CFGR).  The following proposition characterises shear-free ray congruences.

\begin{proposition} \label{prop:SFR}  Let $(M^3, g_t, U_t)$ evolve according to {\rm (\ref{coupled-evol})} and suppose that $\wh{T}(U) = - \frac{1}{2}\grad_gf + \frac{1}{2}U(f)U$\,, so that the null curves generated by $U$ are geodesic.  Then the congruence is shear-free in $(\Mm^4, \Gg = - f^2\dd t^2 + g_t)$ if and only if
$$
f\si (Z, Z) + T(Z, Z) = 0\,,
$$
for any complementary frame $Z$.
\end{proposition}

\begin{proof}  The condition that the ray congruence be shear-free is given by
$$
\left( \Ll_{(\pa_t + fU)}\Gg \right)(Z, Z) = 0\,,
$$
equivalently by
$$
\Gg \left( \na_Z(\pa_t + fU), Z\right) = 0\,.
$$
The condition of the proposition now follows from Lemma \ref{lem:comp}(i).
\end{proof}

In order to generate an SFR, a natural evolution now presents itself, given by the following theorem. Write $\mu = \na^g_UU$ for the mean curvature of the integral curves of $U$ and let $\mu^{\flat} = g(\mu, \ \cdot \ )$ be its dual.

\begin{theorem} \label{thm:SFR} Let $f : I \times M^3 \ra \RR$ $(I\subset \RR$ open, connected) be a non-vanishing function and let $(M^3, g_t, U_t)$ evolve according to the equations:
\begin{equation} \label{coupled-evol-SFR}
\left\{ \begin{array}{lll}
{\rm (i)} \qquad \ds \frac{\pa U}{\pa t} & = & - f\na^g_UU + \grad_gf \\
{\rm (ii)} \qquad \ds\frac{\pa g}{\pa t} & = & - 2\theta \odot \left( 2\dd f\vert_{TM} - f\mu^{\flat} - U(f)\theta \right) - f\Ll_Ug\,.
\end{array}\right.
\end{equation}
Then the null curves generated by $U$ form a shear-free ray congruence.  Furthermore, provided that $M^3$ is compact, for given initial conditions $(g_0, U_0)$, the above system has a solution for small time.
\end{theorem}

\begin{proof}  By the geodesic condition Lemma \ref{lem:props-U}(iv) and the shear-free condition given by Proposition \ref{prop:SFR}, the tensor $T$ in \eqref{coupled-evol} is determined by the conditions:
\begin{eqnarray*}
\wh{T}(U) & = & - \frac{1}{2}\grad_g f + \frac{1}{2}U(f)U \\
 T(Z, Z) & = & - f\si (Z, Z)\,.
\end{eqnarray*}
This implies
$$
T = - \dd f\odot \theta + U(f) \theta^2 + f\theta\odot \mu^{\flat} - \frac{f}{2}\Ll_Ug\,.
$$
Substitution into \eqref{coupled-evol} now gives \eqref{coupled-evol-SFR}.

To prove the existence of solutions, we note that, with respect to a system of local coordinates $(x^i)$ on $M^3$, the system (\ref{coupled-evol-SFR}) is a first order system of differential equations in the coefficients $u^i$ and $g_{ij}$ of $U = u^i \pa_i$ and $g = g_{ij} \dd x^i \dd x^j$, respectively.  Hence, on a compact domain, for given initial conditions, we are assured of a solution for small time by Cauchy's existence theorem, see for example \cite{Di}(10.4.5).
\end{proof}

\begin{remark}  In the above theorem, there is no constraint on the initial vector field $U_0$.  At first sight, this appears to contradict the flat space case described in Section \ref{sec:2}, where each vector field $U_t$ is necessarily tangent to a conformal foliation ($\si \equiv 0$).  However, in general the metric $g_t$ will evolve and produce a space-time with non-trivial curvature.  In fact, if we take $f \equiv 1$, then
$$
\frac{\pa g}{\pa t} = 2\theta \odot \mu^{\flat} - \Ll_Ug
$$
where the right-hand side vanishes if and only if $U$ is tangent to a Riemannian foliation.
\end{remark}

\section{Congruences that generate conformal foliations}

We now address the second generalisation of an SFR on Minkowski space, that of a family of evolving conformal foliations.  This turns out to be much more difficult, but has the advantage of providing a field equation somewhat akin to the Schr\"odinger equation on space-time.

We suppose once more that we have a triple $(M^3, g_t, U_t)$ of a $3$-manifold, a metric and a unit vector field that evolve according to \eqref{coupled-evol}, where $t$ belongs to some connected interval $I\subset \RR$.  For convenience of notation, in what follows, we set $\tau = g(\na_UZ, \ov{Z})$.  The following proposition expresses how $\si$ and $\rho$ evolve along null curves generated by $U$.  

\begin{proposition} \label{prop:prop-conf}  Let $(M^3, g_t, U_t)$ evolve according to {\rm (\ref{coupled-evol})}, then the evolutions of $\si , \rho : I \times M^3 \ra \CC$ along null curves generated by $U$ are given by the formulae:
\begin{eqnarray}
(\pa_t + fU)\si & = &\left( \frac{1}{2}T(Z, \ov{Z}) + f\tau - \frac{f}{2}(\rho + \ov{\rho})\right) \si - \frac{1}{2}T(Z,Z)\big(\rho -2 U(\ln f)\big) \nonumber \\
  & & + Z(\ln f) \left( \frac{Z(f)}{2} - T(U, Z)\right) - \left( \frac{Z(f)}{2} + T(U, Z)\right) g(\na^g_UU, Z) \nonumber \\
 & & \qquad + \Gg ({}^{\Gg}R(\pa_t, Z)U, Z) - f\Ric^g (Z, Z)\,, \label{prop-conf}
\end{eqnarray}
\begin{eqnarray}
(\pa_t + fU)\rho & = & - \frac{1}{2}\big( f\rho + T(Z, \ov{Z})\big)\rho - \frac{1}{2} \left( f \ov{\si} + T(\ov{Z}, \ov{Z})\right) \si + \frac{1}{f}Z(\ln f)\ov{Z}(\ln f) \nonumber \\
  & & \qquad + \frac{1}{2} T(Z, \ov{Z})\tau - \Ric^{\Gg}(\pa_t, U) - f\Ric^g(U, U)\,. \label{prop-rho}
\end{eqnarray}
\end{proposition}

\begin{proof}  Unless stated otherwise, in what follows, the connection $\na$ is the Levi-Civita connection with repect to the metric $\Gg$.  Then
\begin{eqnarray*}
(\pa_t + fU)\si & = & \Gg (\na_{\pa_t + fU}\na_ZU, Z) + \Gg (\na_ZU, \na_{\pa_t + fU}Z) \\
 & = & \Gg ({}^{\Gg}R(\pa_t + fU, Z)U + \na_Z\na_{\pa_t+fU}U + \na_{[\pa_t+fU, Z]}U, Z) + \Gg (\na_ZU, \na_{\pa_t+fU}Z) \\
 & = & \Gg \left( {}^{\Gg}R(\pa_t+fU, Z)U + \na_{[\pa_t+fU, Z]}U, Z\right) + \Gg (\na_ZU, \na_{\pa_t+fU}Z)
\end{eqnarray*}
since by Lemma \ref{lem:props-U}(ii), $\na_{\pa_t+fU}U = 0$.  We now evaluate the various quantities in this expression.

The set of vector fields $\{ Z, \ov{Z}, U, \pa_t/f\}$ form a basis of the complexified tangent space $T^{\CC}(I\times M^3)$, with the properties $\Gg (Z,Z) = \Gg (\ov{Z}, \ov{Z}) = 0$ and $\Gg (Z, \ov{Z}) = 2$; the other vectors being mutually orthonormal.  Thus we may write
$$
\na_{\pa_t}Z = \al_1Z + \be_1 \ov{Z} + \ga_1 U + \delta_1 (\pa_t/f)\,,
$$
for some functions $\al_1, \be_1, \ga_1, \delta_1$, which we calculate as follows:
$$
\begin{array}{l}
\al_1 = \frac{1}{2}\Gg (\na_{\pa_t}Z, \ov{Z})= \frac{1}{2}T(Z, \ov{Z})\ ({\rm Lemma \ \ref{lem:comp}(ii)}), \\
 \be_1 = \frac{1}{2}\Gg (\na_{\pa_t}Z, Z) = 0 \\
\ga_1 = \Gg (\na_{\pa_t}Z,U) = - \Gg (Z, \na_{\pa_t}U) = \Gg (Z, f\na^g_UU)  \\
 \delta_1 = - \Gg \left(\na_{\pa_t}Z, \frac{\pa_t}{f}\right) = \frac{1}{f}\Gg \left( Z, \na_{\pa_t}\pa_t\right) = \frac{1}{f}\Gg (Z, fg^{kl}\pa_lf \pa_k) = \Gg (Z, \grad_gf)\,.
\end{array}
$$
Thus
$$
\na_{\pa_t}Z = \frac{1}{2}T(Z, \ov{Z}) Z + fg(Z, \na^g_UU)U + Z(f)\frac{\pa_t}{f}\,.
$$

Now set 
$$
\na_UZ = \al_2 Z + \be_2\ov{Z} + \ga_2U + \delta_2 (\pa_t/f)\,,
$$
for functions $\al_2, \be_2, \ga_2, \delta_2$, which are given by
$$
\begin{array}{l}
\al_2 = \frac{1}{2}\Gg (\na_UZ, \ov{Z})= \frac{1}{2}\tau , \quad \be_2 = 0 \\
\ga_2 = \Gg (\na_UZ,U) = - \Gg (Z, \na_UU)   \\
 \delta_2 = - \Gg \left(\na_UZ, \frac{\pa_t}{f}\right) = \frac{1}{f}\Gg \left( Z, \na_U\pa_t\right) = -\frac{1}{2}Z(\ln f) + \frac{1}{f}T(U,Z) \ ({\rm Lemma \ \ref{lem:comp}(i)})\,.
\end{array}
$$
Thus
$$
\na_UZ = \frac{1}{2}\tau Z - g(Z, \na^g_UU)U - \frac{1}{2f^2}\Big(Z(f) - 2T(U,Z)\Big) \pa_t\,,
$$
and we obtain
$$
\na_{\pa_t + fU}Z = \frac{1}{2}(T(Z, \ov{Z}) + f\tau )Z + \frac{1}{2f}\Big(Z(f) + 2T(U,Z)\Big) \pa_t\,.
$$
Proceeding in the same way, we obtain
$$
\na_ZU = \frac{1}{2}\ov{\rho} Z + \frac{1}{2}\si \ov{Z} - \frac{1}{2f^2}\Big(Z(f) - 2T(U,Z)\Big) \pa_t\,,
$$
to give
$$
\Gg \left( \na_ZU, \na_{\pa_t + fU}Z\right) = \frac{1}{2}(T(Z, \ov{Z}) + f\tau )\si + \frac{1}{f}\left( \frac{Z(f)^2}{4} - T(U,Z)^2\right)\,.
$$
Similarly, we obtain
$$
[\pa_t + fU, Z] = \frac{1}{2}(T(Z, \ov{Z})+f\tau  - f\ov{\rho})Z - \frac{1}{2}f\si \ov{Z} - \frac{1}{2}Z(f)U - \wh{T}(Z)\,,
$$
where we recall that $\wh{T}$ is the linear mapping determined at each point $x$ by $g(\wh{T}(E), F) = T(E,F)$ for all $E,F \in T_xM^3$. From this we deduce that
$$
\na_{[\pa_t + fU, Z]}U = \frac{1}{2}(T(Z, \ov{Z})+f\tau  - f\ov{\rho}) \na_ZU - \frac{1}{2}f\si \na_{\ov{Z}}U - \frac{1}{2}Z(f)\na_UU - \na_{\wh{T}(Z)}U\,,
$$
and so
\begin{eqnarray*}
\Gg \left(\na_{[\pa_t + fU, Z]}U, Z\right) & = & \frac{1}{2}\big( f\tau  - f(\rho +\ov{\rho})\big) \si - \frac{1}{2}T(Z, Z) \rho  \\
 & & - \left( \frac{Z(f)}{2} + T(U, Z) \right) g(\na_UU, Z)\,.
\end{eqnarray*}
  
On the other hand
\begin{eqnarray*}
\Ric^g (Z, Z) & = & \Ric^g (X, X) - \Ric^g (Y, Y) - 2\ii \Ric^g (X, Y) \\
 & = & g({}^gR(U, X)X, U) + g({}^gR)(Y, X)X,Y) - g({}^gR(U, Y)Y, U) \\
 & & \quad - g({}^gR(X, Y)Y, X) - 2\ii g({}^gR(U, X)Y, U) \\
 & = & - g({}^gR(U, Z)U, Z)\,.
\end{eqnarray*}
On applying Lemma \ref{lem:curv-Gg}, we deduce that
\begin{eqnarray*}
\Gg({}^{\Gg}R(U, Z)U, Z) & = & g({}^gR(U, Z)U, Z) + \frac{1}{4f^2}\Big( \pa_t g(Z, U)^2 - \pa_tg(U, U) \pa_tg(Z, Z)\Big) \\
 & = & - \Ric^g (Z, Z) + \frac{1}{f^2} \left( \left( \frac{Z(f)}{2} - T(U, Z) \right)^2 + U(f) T(Z, Z)\right) \,.
\end{eqnarray*}
Combining the various expressions, now gives equation (\ref{prop-conf}).  The calculation of $(\pa_t + fU)\rho$ proceeds in a similar manner; we omit the details.
\end{proof}

The striking property about equation (\ref{prop-conf}) is that the right-hand side is a quadratic form $q = q(Z, Z)$ in the complex vector $Z$.  Now any two quadratic forms on a 1-(complex) dimensional space are proportional.  In effect, if we replace $Z$ by the vector field $\wt{Z} = e^{\ii \theta} Z$, then $q(\wt{Z}, \wt{Z}) = e^{2\ii \theta} q(Z, Z)$ and similarly for any other quadratic form, so, since $\si = \si (Z, Z)$ is also a quadratic form, provided $\si (Z, Z) \neq 0$, the equation
$$
q(Z, Z) = \ga \si (Z, Z)
$$
is well-defined, independent of the choice of $Z$ and determines a well-defined function $\ga = \ga (t, x)$. However, we cannot make this assertion at points where $\si$ may vanish and so cannot claim in general that $(\pa_t + fU)\sigma$ is proportional to $\sigma$, which would allow us to conclude that if $\si$ vanished on a hypersurface, then it would vanish along integral curves of $\pa_t + fU$.  Instead we seek conditions when $q$ vanishes.  We can make this precise with the following consequence of Proposition \ref{prop:prop-conf} and Lemma \ref{lem:props-U}.  In what follows, we write $\Rr$ for the symmetric curvature tensor on $(M^3, g_t)$ defined by
$$
\Rr  = - f\Ric^g + \sym\Gg ({}^{\Gg}R(\pa_t, \,\cdot \, )U, \, \cdot \, )\,,
$$
where $(\sym S)(X,Y) = \frac{1}{2}\big( S(X,Y) + S(Y,X)\big)$ denotes the symmetrisation of an arbitrary covariant $2$-tensor.  Note that $\Rr$ depends on the ambient space-time manifold $(\Mm^4, \Gg )$ and the vector field $U$.  

\begin{theorem} \label{thm:main}
Let $(M^3, g_t, U_t)$ evolve according to the equations:
\begin{equation} \label{coupled-evol-thm}
\left\{ \begin{array}{lll}
{\rm (i)} \qquad \ds \frac{\pa U}{\pa t} & = & - f\na^g_UU + \grad_gf \\
{\rm (ii)} \qquad \ds\frac{\pa g}{\pa t} & = & - 2 \,\theta \odot \left(\dd f\vert_{TM}\right) + 2T
\end{array}\right.
\end{equation}
where $t$ belongs to a connected interval $I \subset \RR$, $T$ is a symmetric covariant $2$-tensor satisfying 
$$
\wh{T}(U) = - \frac{1}{2}\grad_gf + \frac{1}{2}U(f)U
$$
and where $f : I \times M^3 \ra \RR$ is a smooth function.  Then the integral curves of the vector field $\pa_t + fU$ are null geodesics with respect to the space-time metric $\Gg = - f^2 \dd t^2 + g_t$. If further
\begin{equation} \label{quad}
- \frac{1}{2}T(Z,Z)\big(\rho -2 U(\ln f)\big)+ \frac{1}{f}Z(\ln f)^2  + \Rr (Z, Z)=0\,, 
\end{equation}
where $Z = X - \ii Y$ is a complex vector field tangent to $M^3$ with $X, Y, U$ an orthonormal triple, then if $U_0$ is tangent to a conformal foliation on $(M^3, g_0)$, it remains tangent to a conformal foliation on $(M^3, g_t)$ for all future time $t$.
\end{theorem}
Note that although $T(U, \ \cdot \ )$ is determined, we are free to choose $T$ on the orthogonal complement of $U$ in $TM^3$.  However, $T$ is required to be a \emph{real} tensor and the complex coefficient $\rho = g(\na_{\ov{Z}}U, Z)$ means we cannot solve (\ref{quad}) in its most general form.  In what follows, we describe a number of special cases when solutions to \eqref{coupled-evol-thm} and \eqref{quad} exist.

\medskip

\noindent {\bf Case 1}:  \emph{Stable points of the evolution}.  These are given by the equations:
\begin{equation} \label{stat}
\left\{ \begin{array}{rcl}
\grad_g \ln f  & = & \na^g_UU \\
 T & = & \theta \odot \dd f\vert_{TM} 
\end{array}\right.
\end{equation} 
In particular, we must have $U(f) = 0$ and on writing $\mu = \na^g_UU$ for the mean curvature of the integral curves of $U$, necessarily $\mu$ is a gradient and $T = f\theta \odot \mu^{\flat}$.  If we now make the additional assumption that $U_0$ is tangent to a conformal foliation, then since $g_t$ and $U_t$ are stable (independent of $t$), then $U_t = U_0$ is also tangent to a conformal foliation and $\si$ must vanish identically.  It then follows from \eqref{prop-conf}, equivalently \eqref{quad}, that
$$
\Ric^g(Z, Z) = Z(\ln f)^2\,.
$$
In fact, given a conformal foliation on a Riemannian $3$-manifold, one can calculate the Ricci curvature in terms of invariants of the foliation, see \cite{Ba-Da}, Proposition 6.1.  One then obtains
\begin{eqnarray*}
\Ric^g(Z, Z) & = & \frac{1}{2}(\Ll_{\mu}g)(Z, Z) - \mu^{\flat}(Z)^2 \\
 & = & \na^g\dd \ln f(Z, Z) - Z(\ln f)^2\,.
\end{eqnarray*}
Combining these two expression for $\Ric^g(Z, Z)$ we have a necessary condition on $f$:
$$
\na^g\dd \ln f(Z, Z) = 2Z(\ln f)^2\,.
$$

Conversely, given a vector field $U = U_0$ which is tangent to a conformal foliation with mean curvature $\mu = \na^g_UU = \grad_g f$ a gradient.  Then on defining $T$ by \eqref{stat}, we have a stationary point of the evolution \eqref{coupled-evol-thm}.  To summarize: \emph{stationary points correspond to vector fields $U$ whose associated mean curvature vector field $\mu$ is a gradient; if $U$ is initially tangent to a conformal foliation, then it remains so}.

A special case is when $U$ has integral curves geodesic.  Then we can take $f \equiv 1$ and $T \equiv 0$.  In this case, if the integral curves of $U$ also form a conformal foliation, then locally, $U$ is tangent to the fibres of a harmonic morphism, that is a semi-conformal map which is also harmonic.  It is to be noted that then, necessarily $\Ric^g(Z, Z) = 0$, see \cite{Ba-Wo-6}.  Furthermore, if $U$ is a vector field which is tangent to a conformal foliation, then the property that $\mu$ be a gradient implies that there is a conformal deformation of the metric $g$ which renders the integral curves geodesic (see \cite{Ba-Wo-6}, Lemma 4.6.6).  The converse may not be true.  

\medskip

\noindent {\bf Case 2}:  \emph{The case when $(M^3, g_0)$ has constant curvature}.  In this case, we take $T\equiv 0$ and $f \equiv 1$.  Then by \eqref{coupled-evol-thm}, $\pa g/ \pa t = 0$ and $g_t = g_0$ for all $t$.  Since $\pa \Ga^k_{ij} / \pa t = 0$, by Lemma \ref{lem:curv-Gg}, it follows that $R^0_{ijk} = 0$, so that $\Gg ({}^{\Gg}R(\pa_t, Z, U, Z) = 0$.  On the other hand, by the hypothesis of constant curvature, $\Ric^g = cg$ for some constant $c$ and 
$$
\frac{\pa U}{\pa t} = - \na^g_UU
$$
remains tangent to a conformal foliation. In the flat case, examples can be calculated from the twistor correspondence described in Section \ref{sec:2}.  

\begin{example} \label{ex:flat}  Consider the parametrized surface in $\CP^3$ given by 
$$
(z, w) \mapsto \left[ - \frac{\ii cz}{\sqrt{2}}, w, z, 1\right]\,,
$$
where $c \in \CC$ is non-zero.  The incidence relation $\xi_A = \ii x_{AA'} \eta^{A'}$ takes the form:
$$
\left\{ \begin{array}{rcl}
 - cz & = & vz - \ov{q} \\
 - \ii \sqrt{2} w & = & - qz + u
\end{array} \right.
$$
giving
$$
 z = \frac{\ov{q}}{v + c}\,.
$$
The spinor field tangent to the corresponding SFR is given by 
$$
\mu_A = \rho \left( \begin{array}{c}
 \ov{q} \\
v + c
\end{array} \right)
$$
for some complex-valued function $\rho$.  A semi-conformal map $\phi_t: \RR^3_t \ra \CC$ having the field $U = \frac{1}{1 + |z|^2}(|z|^2 - 1, 2z)$ tangent to its fibres, is given by
$$
\phi_t(x_1, q) = \frac{1}{2\sqrt{2}\ov{q}}\big( |x|^2 - 2x_1(c+t) + (t+c)^2\big)\,.
$$
By Corollary \ref{cor:evol-scm}, $\phi_t$ is uniquely determined up to post-composition with a conformal mapping of the plane $f_t$.  If we take $f_t$ to be the dilation $f_t(w) = \frac{2\sqrt{2}}{t^2}w$, then we obtain the semi-conformal mapping $\psi_t = f_t \circ \phi_t$, with the same fibres of $\phi_t$:
$$
\psi_t(x) = \frac{1}{t^2\ov{q}}\big( |x|^2 - 2x_1(c+t) + (t+c)^2\big)\,. 
$$  
In the limit as $t \ra \infty$, this converges pointwise to the semi-conformal mapping $\psi_{\infty} = \frac{1}{\ov{q}}$ having geodesic fibres.  This shows that in the limit as $t \ra \infty$, the corresponding foliation given by the integral curves of $U$ converges to a conformal foliation by geodesics.  This is as we would expect when convergence as $t \ra \infty$ occurs, since such a foliation is stable with respect to the evolution \eqref{coupled-evol-thm}.
\end{example}

\begin{example}  For non-zero integers $k$ and $l$, there is a smooth semi-conformal map $\phi : S^3 \ra S^2$ between Euclidean spheres of the form
$$
\phi (\cos (s) e^{\ii a}, \sin (s) e^{\ii b}) = (\cos \al (s), \sin \al (s) e^{\ii (ka + lb})\,,
$$
where $\al : [0, \pi / 2] \ra [0, \pi ]$ satisfies $\al (0) = 0, \, \al (\pi / 2) = \pi$ \cite{Ba-Wo-6}, Example 13.5.3.  If we use coordinates $(x^1, x^2, x^3) = (s, a, b)$ on $S^3$, then the canonical metric has the form
$$
g = \dd s^2 + \cos^2s \, \dd a^2 + \sin^2s \, \dd b^2\,.
$$
With respect to the basis $\pa / \pa s,\, \pa / \pa a,\, \pa /\pa b$ of the tangent space, the non-zero Christoffel synbols are given by
$$
\Ga^1_{22} = \sin s\, \cos s , \quad \Ga^1_{33} = - \sin s \, \cos s, \quad \Ga^2_{12} = - \tan s,\quad \Ga^3_{13} = \cot s\,,
$$
and the unit vector field tangent to the fibres of $\phi$ is given by
$$
U_0 = \frac{1}{\sqrt{k^2 \sin^2 s + l^2 \cos^2 s}}\left( - l\pa_a + k\pa_b\right)\,.
$$
Let us suppose that for each $t$, the vector field $U_t$ has the form $U_t = u\pa_s + v\pa_a + w\pa_b$ for functions $u,v,w$ which depend on $s$ and $t$ only.  Then a calculation shows that the equation $\ds \frac{\pa U}{\pa t} = - \na^g_{U_t}U_t$ is equivalent to the system:
$$
\left\{ \begin{array}{rcl}
\ds\frac{\pa u}{\pa t} & = & \ds - (v^2 - w^2) \sin s\, \cos s - u \frac{\pa u}{\pa s} \\
\ds\frac{\pa v}{\pa t} & = & \ds - u\frac{\pa v}{\pa s} + 2uv \, \tan s \\
\ds\frac{\pa w}{\pa t} & = & \ds - u\frac{\pa w}{\pa s} - 2uw \, \cot s\,,
\end{array}
\right.
$$
with initial conditions:
$$
u(s, 0) = 0, \quad v(s, 0) =  \frac{-l}{\sqrt{k^2 \sin^2 s + l^2 \cos^2 s}}, \quad w(s, 0) = \frac{k}{\sqrt{k^2 \sin^2 s + l^2 \cos^2 s}}\,.
$$
The theory of partial differential equations of first order guarantees the existence of a solution for small time.  Note that, even though $U_0$ is tangent to the fibres of the map $\phi : S^3 \ra S^2$, for later time, we cannot expect that there will be a \emph{globally defined} semi-conformal mapping on $S^3$ with $U_t$ tangent to its fibres.  We can only be sure that locally such a mapping exists.
\end{example}

\noindent {\bf Case 3}:  \emph{The case of integrable complementary distribution}. 
If, via the flow \eqref{coupled-evol-thm}, we evolve through vector fields tangent to a conformal foliation, then remarkably, the integrability of the complementary distribution is preserved, as described by the following lemma.

\begin{lemma} \label{lem:im-rho} Let $(M^3, g_t, U_t)$ evolve according to \eqref{coupled-evol-thm}, then
$$
(\pa_t + fU) (\rho - \ov{\rho}) = -\frac{1}{2}T(Z, \ov{Z})(\rho - \ov{\rho}) - \frac{1}{2}T(\ov{Z}, \ov{Z})\si + \frac{1}{2}T(Z,Z)\ov{\si}\,.
$$
In particular, if $\rho - \ov{\rho} = 2\ii g(U, [X, Y])$ vanishes at $t = 0$ and $\si \equiv 0$, then $\rho - \ov{\rho} = 0$ for future time.
\end{lemma}

\begin{proof}  The formula is a direct consequence of \eqref{prop-rho}, under the assumption of Theorem \ref{thm:main}, that $\wh{T}(U) = - \frac{1}{2}\grad_gf + \frac{1}{2}U(f)U$.  The conclusion then follows as a consequence of the uniqueness of solutions of a first order ODE for given initial conditions.
\end{proof}

On taking $f\equiv 1$, this enables us to give a satisfactory form of the evolution \eqref{coupled-evol-thm} in the case when initially the complementary distribution is integrable.

\begin{theorem}  \label{thm:integrable}  Let $(M^3, g_t, U_t)$ evolve according to the equations:
\begin{equation} \label{coupled-evol-int}
\left\{ \begin{array}{lll}
{\rm (i)} \qquad \ds \frac{\pa U}{\pa t} & = & - \na^g_UU  \\
{\rm (ii)} \qquad \ds\frac{\pa g}{\pa t} & = & -\frac{4}{\nu}\left\{ \Rr - 2\Rr(U, \, \cdot \, )\odot \theta + \Rr (U, U) \theta^2\right\}\,,
\end{array}\right.
\end{equation}
where $\nu = g(U, \na^g_XX + \na^g_YY)$ is the mean curvature of the horizonal distribution $(U^{\perp})$ which we assume non-zero at all points of $M^3$, and where $\Rr = - \Ric^g + \sym\Gg ({}^{\Gg}R(\pa_t, \,\cdot \, )U, \, \cdot \, )$\,.  Then if $U_0$ is tangent to a conformal foliation on $(M^3, g_0)$ with integrable complementary distribution $U^{\perp}$, it remains tangent to a conformal foliation on $(M^3, g_t)$ for all future time $t$ for which $\rho$ remains everywhere non-zero.  Furthermore, the complementary distribution remains integrable and the integral curves of the vector field $\pa_t + fU$ are null geodesics with respect to the space-time metric $\Gg = - f^2 \dd t^2 + g_t$. 
\end{theorem}
 
\begin{proof}  The result is a consequence of Theorem \ref{thm:main}, on taking $f \equiv 1$ and $T$ to be the tensor 
$$
- \frac{2}{\rho}\left( \Ric^g + \sym\Gg ( {}^{\Gg}R(\pa_t,\, \cdot \, )U, \, \cdot \, )\right) = - \frac{2}{\rho}\Rr\,,
$$
normalised in order that $\wh{T}(U) = 0$. 
For, with these assumptions, the requirement \eqref{quad} is satisfied.  However, we require that $T$ remain \emph{real}, equivalently, that $\rho$ remain real.  Now by \eqref{prop-conf}, Lemma \ref{lem:im-rho}, and our choice of $f$ and $T$, we have the coupled system:
$$
\left\{ \begin{array}{rcl}
(\pa_t+fU)\si & = & \left( \frac{1}{2}T(Z, \ov{Z}) + \tau - \frac{1}{2}(\rho + \ov{\rho})\right) \si \\
(\pa_t + fU) (\rho - \ov{\rho}) & = & -\frac{1}{2}T(Z, \ov{Z})(\rho - \ov{\rho}) - \frac{1}{2}T(\ov{Z}, \ov{Z})\si + \frac{1}{2}T(Z,Z)\ov{\si}\,,
\end{array} \right.
$$
with initial conditions $\si = 0$ and $\rho - \ov{\rho} = 0$.  Once more, by the uniqueness of solutions of a first order ODE for given initial conditions, it must be the case that $\si$ and $\rho - \ov{\rho}$ remain zero for future time.  We now set $\nu$ equal to the real quantity $- \rho$.
\end{proof} 

\begin{remark} We can deal with the case when $\rho \equiv 0$ in a similar way, but now choosing $f$ to be a function such that $U(f)$ is nowhere vanishing (where this is possible).  Then equation \eqref{quad} requires that we take $T$ to be the tensor
$$
- \frac{1}{U(f)} \left\{ \frac{1}{f}(\dd\ln f)^2 + \Rr\right\}
$$
normalized so that $\wh{T}(U) = 0$.
\end{remark}

It is interesting that \eqref{coupled-evol-int}(ii) is a modified Ricci flow (depending on the influence of the term $\Gg ({}^{\Gg}R(\pa_t, \, \cdot \, )U, \, \cdot \,)$ in $\Rr$), coupled to the evolution of the vector field $U$.  Some examples of Ricci solitons evolve under the Ricci flow through semi-conformal mappings \cite{Ba-Da}.  To prove existence of solutions to the system \eqref{coupled-evol-int}, even for short time, would seem to be quite tricky.  The term $\Gg ({}^{\Gg}R(\pa_t, \, \cdot \, )U, \, \cdot \,)$ that appears in the tensor $\Rr$ in general contains terms involving the derivative of the metric $g$; this will affect the nature of the equations.  The following example gives an indication of how this might happen.  Note that the example, consisting of a conformal foliation by \emph{geodesics}, also fits into Case 1, and with the choice $f \equiv 1$, $T \equiv 0$, would correspond to a stationary evolution.  However, with the imposition of the evolution of the metric given by Theorem \ref{thm:integrable}, the metric will evolve in time.

\begin{example}  We consider a radial metric on a domain of $\RR^3$ of the form
$$
g = a(t, r)^2 \dd r^2 + b(t, r)^2 g_{S^2} = a^2 \dd r^2 + b^2 (\dd s^2 + \sin^2s\, \dd \theta^2)\,,
$$
where the coordinates correspond via $(x^1, x^2, x^3) = (r,s, \theta )$, with $(s, \theta )$ standard coordinates on $S^2$.  Then the $r$-curves form a conformal foliation by goedesics with unit tangent given by $U = \frac{1}{a}\frac{\pa}{\pa r}$.  The non-zero Christoffel symbols are given by
\begin{eqnarray*}
\Ga^1_{11} = \pa_r \ln a,& \  \Ga^1_{22} = - \frac{b\pa_rb}{a^2}, \ \Ga^1_{22} = - \frac{b\sin^2s\,\pa_r b}{a^2},\ \Ga^2_{12} = \pa_r \ln b , \ \Ga^3_{13} = \pa_r \ln b, \\
&  \Ga^2_{33} = - \sin s \cos s , \ \Ga^3_{23} = \cot s\,.
\end{eqnarray*}
The tensor $\Rr$ then has the following non-zero components:
\begin{eqnarray*}
\Rr_{11} & = & 2\left(\pa_{rr}{}^2\ln b - \pa_r\ln a\, \pa_r\ln b + (\pa_r\ln b)^2\right) \\
\Rr_{13} & = & \cot s\,\pa_r\ln b \\
\Rr_{22} & = & - 1 + \frac{1}{a}\pa_r \left( \frac{b\pa_rb}{a}\right) + \frac{b}{a}\left( \pa_{tr}{}^2b - \pa_t\ln a\, \pa_r b\right) \\
\Rr_{33} & = & \sin^2s \, \Rr_{22}\,.
\end{eqnarray*}
An orthonormal triple containing $U$ is given by $\{ X, Y, U\} = \{ \frac{1}{b}\frac{\pa}{\pa s}, \frac{1}{b\sin s}\frac{\pa}{\pa \theta}, U\}$.  A calculation now gives
$$
\nu = g(U, \na^g_XX + \na^G_YY) = - \frac{2}{ab}\pa_rb\,.
$$
The evolution equation \eqref{coupled-evol-int} is then just the single equation:
$$
b\pa_{tr}{}^2 b - \pa_tb\,\pa_rb - b\pa_t\ln a\, \pa_rb = a - \pa_r \left( \frac{b\pa_rb}{a}\right)\,.
$$
Stationary points, i.e. time-independant solutions require
$$
 a - \pa_r \left( \frac{b\pa_rb}{a}\right) = 0\,.
$$
But now we can suppose that $r$ is parametrised so that $a(r) \equiv 1$.  Then up to replacing $r$ by $r+c$ for a constant $c$, we have $b = r$, which corresponds to the Euclidean metric.
\end{example}

\section{The corresponding field equations and conformal invariance} \label{sec:final} We conclude with some general remarks concerning the application of the above to the description of massless fields.  Our perspective has been to couple the evolving field, which we suppose represented by the vector field $U$ (or the corresponding semi-conformal map $\phi$), with the gravitational field $g$.  Thus the two should evolve simultaneously following some well-defined equations.  In order to arrive at these equations, we have adapted the flat case (Minkowski space) to more general curved metrics.  From the physical point of view, we do not wish to ``presuppose" the space-time manifold $(\Mm^4, \Gg )$, but rather to introduce a natural dynamic into the pair $(g_t, U_t)$ based on the wish to preserve the conformality of the foliation and the geodesity of the null curves determined by $U$.  As such, we would like to be able to see semi-conformal mappings as deriving from a variational principle.  In fact, this is the case if the mapping has values in a surface.

Let $\phi : (M^m, g) \ra (N^2, h)$ ($m\geq 2$) be a smooth mapping of a Riemannian manifold into a Riemannian surface.  Then there are at most two non-zero eigenvalues $\la_1{}^2, \la_2{}^2$ of the first fundamental form $\phi^*h$, which we can suppose satisfy $\la_1 \geq \la_2 \geq 0$.  Then $\phi$ is semi-conformal if and only if $\la_1 = \la_2$ at each point.  On the other hand, the \emph{energy density} is the function:
$$
e_{\phi} = \frac{1}{2}\Tr (\dd \phi \circ \dd \phi^*) = \frac{1}{2}(\la_1{}^2 + \la_2{}^2)\,,
$$
where $\dd\phi^*$ denotes the adjoint of the derivative, and the \emph{Jacobian} the function:
$$
j_{\phi} = \sqrt{\det ( \dd \phi \circ \dd \phi^*)} = \la_1 \la_2\,.
$$
Clearly $e_{\phi} - j_{\phi} \geq 0$ with equality if and only if $\la_1 = \la_2$, i.e. if and only if $\phi$ is semi-conformal.  If $M^m$ is now compact (otherwise we restrict to compact domains), it follows that $\phi$ is semi-conformal if and only if it is absolutely minimizing with value zero for the functional:
$$
I(\phi ) = \int_M(e_{\phi} - j_{\phi})^{m/2} dv_g\,,
$$
where $dv_g$ is the volume element on $M^m$ and where the power $m/2$ assures the conformal invariance of $I$ with respect to conformal deformations of $g$.  It is not known whether in any given homotopy class, the functional $I$ admits an absolute minimum and if so whether it is always zero. 

We can obtain a time-dependent differential equation governing the evolution of $\phi$ in a similar way to Corollary \ref{cor:evol-scm} for the flat case.  On differentiating the equation $\dd \phi_t(U_t)$ with respect to $t$ and applying \eqref{coupled-evol-thm}(i), we obtain
\begin{eqnarray*}
\dd \left(\frac{\pa \phi_t}{\pa t}\right) (U) & = & - \dd \phi_t(U_t) = \dd\phi_t(\na^g_UU) - \dd \phi_t(\grad_gf)\\
 & = & - \tau (\phi_t) - \dd \phi_t(\grad_gf)\,;
\end{eqnarray*}
the latter equality coming from the fundamental equation \eqref{fund-scm} for a semi-conformal map.  We can once more integrate this along integral curves of $U$ to obtain:
$$
\frac{\pa \phi_t}{\pa t}(x) = - \int_{s = 0}^{s_0} \left\{\tau (\phi_t) + \dd \phi_t(\grad_gf)\right\} (c(s))\,\dd s\,,
$$
where the start point $c(0)$ lies on a chosen hypersurface of $M^3$ which the integral curves of $U$ hit transversally and $x = c(s_0)$.  We can view this equation as a modified Schr\"odinger equation, or heat equation, with potential given by the function $f$. 

A further observation is as follows.  If $U$ is tangent to a conformal foliation on a $3$-manifold $(M^3, g)$, now possibly with singularities, and $\phi : M^3 \ra N^2$ is a semi-conformal map with $U$ tangent to its fibres, then where it is smooth, the vector field $\la^2 U$ is divergence free, where $\la$ denotes the dilation of $\phi$.  In particular, if $S\subset M^3$ is a closed surface which avoids singularities of $U$ (and of $\phi$), then by Stokes Theorem, the quantity:
$$
Q(\Ss ) = \int_S\la^2g(U, n)\dd v_{S}\,,
$$
is independent of the homology class of $S$ in $M^3 \setminus \Si$, where $\Si$ represents the singular set of $U$ \cite{We}.  Furthermore, it is easily checked that the integral $Q(\Ss )$ is conformally invariant with respect to conformal deformations of the metric $g$.  For example, if $\phi : \RR^3 \ra S^2$ is radial projection given by $x \mapsto x/|x|$, with a singularity at the origin and $S$ is taken to be any sphere surrounding the origin, then $Q(S ) = 4\pi$.  We then expect $Q(S )$ to represent some intrinsic quantity associated to the field given by $\phi$.

One of the features of a shear-free ray congruence on Minkowsi space is its conformal invariance.  Thus if $\mu^A$ satisfies equation \eqref{SFR} with respect to the Minkowski metric $g_{\MM^4}$, then it does so with respect to any conformally equivalent metric $\Om^2g_{\MM^4}$, where $\Om : \MM^4 \ra \RR$ is a smooth non-vanishing function.  Our equations \eqref{coupled-evol-SFR}, \eqref{coupled-evol-thm} and \eqref{coupled-evol-int} do not have this general invariance property.  This is because of our requirement that the vector field $U$ be parallel along null curves generated by $U$, which is not a conformally invariant property.  Explicity, if we consider a conformal deformation $\wt{\Gg} = \Om^2 \Gg$ of the space-time metric $\Gg = - f^2 \dd t^2 + g_t$, then for arbitrary vector fields, we have the relation:
$$
\na^{\wt{\Gg}}_EF = \na^{\Gg}_EF + E(\ln \Om )F + F(\ln \Om )E - \Gg (E,F) \grad_{\Gg}\ln \Om\,.
$$
If we now consider the new quantities $\wt{U} = U/\Om$, $\wt{f} = \Om f$, then
$$
\na^{\wt{\Gg}}_{\pa_t + fU} \wt{U} = \frac{1}{\Om}\na^{\Gg}_{\pa_t + fU}U + \frac{1}{\Om}U(\ln \Om )(\pa_t + fU) - \frac{f}{\Om}\grad_{\Gg}\ln \Om\,,
$$
so that both $U$ and $\wt{U}$ are parallel along the null curves they generate (with respect to $\Gg$ and $\wt{\Gg}$, respectively) if and only if 
$$
\left\{ \begin{array}{rcl}
(\pa_t + fU)(\Om ) & = & 0 \\
 U(\Om )U - \grad_g\Om & = & 0\,. 
\end{array} \right.
$$
The second equation is quite restrictive.  It says that the gradient of $\Om$ must be parallel to $U$, but then if $\Om$ is non-constant, this means that the complementary distribution $U^{\perp}$ is integrable.  In particular, if $U^{\perp}$ is not integrable, then the only conformal deformation possible has $\Om$ constant.  On the other hand, if $U^{\perp}$ is integrable, then any smooth nowhere vanishing function $\Om$ which is constant along the integral surfaces of the distribution $U^{\perp}$ which also satisfies $(\pa_t + fU)(\Om ) = 0$ ($\Om$ is constant along the null curves generated by $U$) is an allowable conformal deformation.  It is easily checked that the equations \eqref{coupled-evol-SFR}, \eqref{coupled-evol-thm} and \eqref{coupled-evol-int} are invariant with respect to such conformal deformations. 

However, we take the view that such a \emph{global conformal invariance} should not be a requirement of our model.  The concept of \emph{conformal foliation}, or \emph{semi-conformal map}, is a conformal invariant, so we can modify our initial conditions $(M^3, g_0, U_0)$ by any conformal deformation: $(M^3, \wt{g}_0 = \om^2 g_0, \wt{U}_0 = \frac{1}{\om}U_0)$ and then allow the approprate evolution to take place.  This is consistent with our perspective that the triple $(M^3, g_0, U_0)$ has an implicit dynamic associated to it.  We could then view the conformal invariance of the initial conditions as a \emph{local conformal invariance}.

\end{document}